\documentclass[11pt,a4paper]{amsart}
\usepackage{ulem} 
\usepackage{times,epsfig}
\usepackage{amsfonts,amscd,amsmath,amssymb,amsthm}%

\usepackage{xcolor}
\numberwithin{equation}{section}

\DeclareMathOperator{\tr}{Tr}
\DeclareMathOperator{\Ran}{Ran}

\DeclareMathOperator{\sign}{sign}

\newcommand{\R}{\mathbb{R}}

\newcommand{\Z}{\mathbb{Z}}

\newcommand{\M}{\mathbb{M}}

\newcommand{\1}{\mathbb{I}}

\newcommand{\bS}{\mathbb{S}}

\newcommand{\cK}{{\mathcal K}}
\newcommand{\cL}{{\mathcal L}}
\newcommand{\cM}{{\mathcal M}}

\newcommand{\cP}{{\mathcal P}}

\newcommand{\cR}{{\mathcal R}}

\newcommand{\cU}{{\mathcal U}}

\newcommand{\rd}{\mathrm{d}}
\newcommand{\e}{\mathrm{e}}
\newcommand{\rv}{\mathrm{v}}

{\bf}{\it}
{\bf}{\it}
{\bf}{\it}
{\bf}{\it}
{\bf}{\it}

\newtheorem{theorem}{Theorem}[section]{\bf}{\it}
\newtheorem{proposition}[theorem]{Proposition}{\bf}{\it}
\newtheorem{corollary}[theorem]{Corollary}{\bf}{\it}
\newtheorem{example}[theorem]{Example}{\it}{\rm}
\newtheorem{lemma}[theorem]{Lemma}{\bf}{\it}
\newtheorem{remark}[theorem]{Remark}{\it}{\rm}
\newtheorem{definition}[theorem]{Definition}{\bf}{\it}
{\bf}{\it}
{\bf}{\it}
{\bf}{\it}

\title{Piecewise linear manifolds: \\Einstein metrics and Ricci flows}

\author[R. Schrader]{Robert Schrader}
\address{Robert Schrader\\ Institut f\"{u}r
Theoretische Physik, Freie Universit\"{a}t Berlin, \newline Arnimallee 14,
 D-14195 Berlin, Germany}
\email{robert.schrader@fu-berlin.de}
\date{\today File: \textbf{PLEEv2.tex}}

\dedicatory{Dedicated to Ludwig Faddeev on the occasion of his 80th birthday}

\begin{document}

\begin{abstract}
This article provides an attempt to extend concepts from the theory of Riemannian manifolds to 
piecewise linear spaces. In particular we propose an analogue of the Ricci tensor, 
which we give the name of an {\it Einstein vector field}. 
On a given set of piecewise linear spaces we define and discuss (normalized) Einstein flows. 
Piecewise linear Einstein metrics are defined and examples 
are provided. Criteria for flows to approach Einstein metrics are formulated. Second variations of 
the total scalar curvature at a specific Einstein space are calculated. 
\end{abstract}

\maketitle

\section{Introduction}\label{sec:intro}
As may be less known, piecewise linear (p.l.) spaces share many of the properties of Riemannian manifolds. 
The first to observe this was Regge \cite{Regge}, who gave a definition of the analogue of the total scalar curvature. 
Therefore sometimes one speaks of {\it Regge calculus}, when discussing p.l. spaces. 
In \cite{CMS2} further curvatures
like Lipschitz-Killing curvatures and boundary curvatures were introduced and their relation to the corresponding smooth 
partners established. A consequence was a new proof of the Chern-Gauss-Bonnet theorem. 
The interest in physics arose from the proposal to use Regge calculus as an approach to quantum gravity in analogy to 
lattice gauge theories \cite{CMS1,F,RoWi}. For this the names {\it lattice gravity} or 
{\it simplicial gravity} is often used, for overviews see e.g. \cite{Hamber,ReggeWilliams}. 
Although Regge worked in a context 
which was purely classical, it was Wheeler, who speculated on the possibility of employing Regge calculus as 
a tool for constructing a quantum theory of gravity \cite{Wheeler}. 
More recently attempts have been made to introduce additional curvature notions. In particular 
analogues of the 
Ricci tensor and a Ricci flow 
\cite{AlsingMcDonaldMiller,CGY,ChowLuo,HW,JinKimGu,Luo,MMcDAlGuYau,MACR,WYZGCTY,ZengSamarasGu}, 
a Yamabe flow  \cite{Glickenstein}, as well as an analogue of an 
Einstein space were proposed \cite{CGY}.

The main motivation for this article is to provide new instruments and insights in the theory of p.l. spaces.
We focus on providing analogues of
\begin{itemize}
 \item{the Ricci tensor,}
 \item{a smooth Einstein space,}
 \item{a (normalized) Ricci flow,}
 \end{itemize}
and we study their properties.
Actually two alternative definitions of analogues of the Ricci tensor and of an Einstein space are given. 
As far as we understand these 
definitions differ from the proposals made so far with the exception of one in \cite{CGY} and we 
shall comment on this below.

We will make a great effort to point out analogies between concepts and quantities appearing in the theory
of p.l. spaces and those showing up in Riemannian manifolds, which we often will
call the {\it smooth case}. 

For short, a p.l. space is obtained by gluing {\it euclidean simplexes} together.
Thus given a p.l. space in this form, its data are given by a simplicial complex plus the lengths of its 
edges, which have to satisfy certain conditions extending the triangle inequalities. The collection of the 
(squared) edge lengths will be called a metric. As for the analogue of the Ricci tensor our definition is motivated 
by the well 
known fact that in the smooth case the Ricci tensor is obtained from the variation of the total scalar curvature.
Analogously the metric is recovered from the volume. Thus we define the Ricci vector field as the gradient 
(with respect to the metric) of the total scalar curvature. For an Einstein space by definition the Ricci vector field is 
proportional either to the metric or to the gradient of the volume.
Introducing the notion of the (normalized) Ricci flow is then straightforward.

This article is organized as follows. In Section \ref{sec:intro} we recall the basic notations and notions in 
the theory of piecewise
linear spaces. It starts with the notion of a pseudomanifold, the analogue of a smooth manifold. Then we introduce 
the notion of a 
metric, with the help of which one can define the volume and the total scalar curvature. There we also define the 
Einstein vector field, see Definition \eqref{ricci}, 
and which may also be written in the equivalent form \eqref{totalscalar3}.
Section \ref{sec:allmetrics} provides a characterization of the space of all metrics on a given pseudomanifold, 
collected in Theorem \ref{theo:5}.
In Section \ref{sec:einsmet} we define Einstein metrics, actually there are two possible definitions 
(as already mentioned), 
see Definition \ref{def:einstein}. As in the smooth case (see e.g. \cite{Besse}) there are equivalent 
conditions for a metric to be Einstein, see Theorems \ref{theo:EI} and \ref{theo:EII}.
Examples of Einstein spaces are provided which are the analogues of $n$-spheres and $n$-tori. 
In Section \ref{sec:ricciflows} an Einstein flow and 
two normalized Einstein flows are defined. These two definitions are closely related to the two definitions 
of an Einstein metric. These normalized flows are such that Einstein metrics are fixed points. Moreover under these 
flows the total scalar curvature always decreases away from Einstein metrics, see Theorems \ref{theo:4} and 
\ref{theo:6}. In Section \ref{sec:var} we discuss the behavior of the total scalar curvature near a special 
Einstein space by computing the second variation under the constraint that either the fourth moment of the edge 
lengths or the volume stays fixed. In the first case the second variation is negative definite, 
in the second case it is indefinite and non-degenerate, see Theorems \ref{theo:locmax} and \ref{theo:volvar}.
Section \ref{sec:op} provides a list of open problems.

For the purpose of comparison with the smooth case, in Appendix \ref{app:Einstein} we recall some well known 
facts from Riemannian geometry. In particular we provide an extensive discussion of the behavior 
of many quantities like 
the scalar curvature, the total scalar curvature and the volume under a scaling of the metric. 
In Appendix \ref{app:1} 
the volume and the total curvature of p.l analogues of $n-$spheres are calculated. Appendix \ref{app:2} 
establishes among other things
smoothness properties of the volume and the total scalar curvature as a function of the p.l. metric. 
Appendices \ref{app:3} and \ref{app:Dervol} give the proofs of relations needed for 
Theorems \ref{theo:locmax} and \ref{theo:volvar}.  

Partial results were presented at the conference in honor of L. Faddeev's 80th birthday, 
{\it Mathematical Physics: Past, present and future}, Euler Institute, St. Petersburg, March 2014.

\textbf{Acknowledgements.} The author thanks J. Cheeger, K. Ecker, H.W. Hamber, and W. M\"{u}ller for extensive and 
helpful discussions and for providing relevant references. 
M. Karowski and A. Knauf carried out the computer calculations. 
A. Knauf has done a wonderful job with a careful proof reading of the entire manuscript. Thanks go to two anonymous 
referees for their helpful and critical remarks.

\tableofcontents


\section{Basic concepts and notations}\label{sec:basic}

For the convenience of the reader we recall basic definitions and properties of the objects we 
will be dealing with (see e.g. \cite{Spanier} and \cite{CMS2}).

A {\it finite simplicial complex} $K$ consists of a finite set of elements called {\it vertices}
and a set of finite nonempty subsets of vertices called {\it simplexes} such that
\begin{enumerate}
\item{ Any set containing only one vertex is a simplex.}
\item{Any nonempty subset of a simplex is also a simplex.}
\end{enumerate}
A $j$-simplex will generally be denoted by $\sigma^j$. The {\it dimension} $j$ is the number of its
vertices minus 1. The 1-simplexes are called {\it edges}. 
If $\sigma^\prime\subset \sigma$, 
then $\sigma^\prime$ is called a {\it face}
of $\sigma$ and a {\it proper face} if $\sigma^\prime\neq\sigma$. 
We set $\dim K = \sup_{\sigma\in K}\dim \sigma$ and occasionally we shall write $K^n$ with 
$\dim K=n$, if we want to emphasize the dimension of $K$. 
A complex $L$ is called a
{\it subcomplex} of $K$ if the simplexes of $L$ are also simplexes of $K$. We write $L \subseteq K$. 
The $k$-skeleton $\Sigma^k(K^n)$ of $K^n\; ( 0 \le k < n)$ is the subcomplex formed by the $j$-simplexes 
with $j\le k$. It is not necessarily a pseudomanifold (for the definition, see below).

In order not to burden the notation, we often will also use $\sigma^j$ to denote the simplicial complex 
formed by this $j$-simplex and all its faces. Also we will use $1-$ simplexes as indices for coordinates, 
such as a point $\underline{z}$ in some euclidean space and then $\partial^{\sigma^1}$ stands for
$$
\frac{\partial}{\partial z_{\sigma^1}}.
$$
The {\it Euler characteristic} of $K$ is defined to be
\begin{equation*}
 \chi(K)=\sum_{k}(-1)^k\sharp(\sigma^k).
\end{equation*}
Let $p=\{p_j\mid 1\le j\le q+1\}$ be points in the euclidean space $E^n,\, n>q$ , which lie in no $(q-1)$-
dimensional affine subspace. The convex hull, $\bar{\sigma}(p)$ and its interior $\sigma(p)$ are called
closed and open {\it linear simplexes}, respectively. By regarding $p_j =v_j$ as vectors, we
have
\begin{equation*}
\sigma^q(p)=\big\{\sum_{j}x_jv_j\big\},
\end{equation*}
where $\{x_j\}$ consists of $(q + 1)$-tuples with $x_j > 0$ and
\begin{equation*}
\sum_{j}x_j=1.
\end{equation*}
$\{x_i\}$ are called the  barycentric coordinates of $\sum_{j}x_jv_j$. They are independent of 
the choice of origin
in $E^n$. A map from $\sigma^q(p)$ to $\sigma^q(p')$ which preserves barycentric coordinates is called
linear.

If $e_1,\cdots,e_n$ are the standard basis vectors in $E^n$, their convex hull is called the
{\it standard} (closed) simplex $\sigma(n)$. To any finite simplicial complex with $n$ (ordered)
vertices, we associate a closed subset $^sK$ of $\sigma(n)$, called the {\it geometric realization}
of $K$. Namely, to each simplex $\sigma^i\in K$ with vertices $\sigma^0_{j_1},\cdots, 
\sigma^0_{j_{i+1}}$, we associate the
open linear simplex determined by $e_{j_1},\cdots,e_{j_{i+1}}$. The union of these linear simplexes
is $^sK$. There is a natural metric space structure, the {\it standard metric} on $^sK$, where the
distance between two points $p, q$ is defined as the infimum of the length of all 
piece-wise smooth paths between $p$ and $q$. More generally, we consider metrics on $^sK$
such that any simplex $^s\sigma\subset {^sK}$ with its induced metric is linearly isometric to some
linear simplex. In what follows we shall use $K$ and $^sK$ interchangeably.

The space $K$, equipped with a metric of the above type is called a
{\it triangulated piecewise flat space} (or p.l. space). 

Clearly, any such space is
determined up to isometry by the edge lengths $l_{\sigma^1}$, the distances between the vertices of
1-simplexes $\sigma^1$. In Section \ref{sec:allmetrics} where we discuss the set of all metrics, 
we will see that it is more appropriate to consider the squares of the edge lengths. Moreover, 
there is a closer analogy with a Riemannian 
metric $g$ than there would be with the set of lengths. 
That is we will work with
\begin{equation*}
 z_{\sigma^1}=l_{\sigma^1}^2
\end{equation*}
and we write $(K,\underline{z})$ for a p.l. space to emphasize the dependence on
the collection of the squares of the edge lengths $\underline{z}=\{z_{\sigma^1}\}_{\sigma^1\in K}=
\{l_{\sigma^1}^2\}_{\sigma^1\in K^n}$. Also we shall say that 
$K$ carries the {\it p.l. metric} $\underline{z}$. Here and in what follows, we assume that the $1-$simplexes 
of $k$ are ordered in some way, such that we can view $\underline{z}$ as an element in $\R^{n_1(k)}_+$ and therefore 
also of $\R^{n_1(k)}$. All results will be independent of the particular choice of the ordering. Sometimes we will also 
choose another ordering, when we consider the $1-$simplexes contained in a given $k$-simplex.

For given $\underline{z}$ we denote by 
$|\sigma^j|=|\sigma^j|(\underline{z}),\;j\ge 1,$
the euclidean $j$-volume of the euclidean $j$-simplex to which $\sigma^j$  is linearly
isometric by assumption. In particular $|\sigma^1|=l_{\sigma^1}=\sqrt{z_{\sigma^1}}$. 
For a vertex we set $|\sigma^0|=1$. Below we shall recall a more explicit expression of $|\sigma^j|$
in terms of the $z_{\sigma^1}$ with $\sigma^1\subseteq \sigma^j$, see \eqref{amatrixvol}.
The scaling law
\begin{equation}\label{scaling1}
 |\sigma^j|(\lambda \underline{z})=\lambda^{j/2}|\sigma^j|(\underline{z}),\qquad \lambda >0
\end{equation}
is obvious.

A {\it subdivision} of a p.l. space $(K,\underline{z})$ is a p.l. 
space $(K^\prime,\underline{z}^\prime)$ 
and a homeomorphism
$$
s : (K^\prime,\underline{z}^\prime)\quad\rightarrow\quad(K,\underline{z})
$$
with the following properties
\begin{itemize}
\item{For every simplex $\sigma^\prime$ in $K^\prime$, 
its image $s(\sigma^\prime)$ is contained in some simplex $\sigma$
of $K$, and $s|\sigma^\prime$ is linear.}
\item{The metric $\underline{z}^\prime$ on $(K^\prime,\underline{z}^\prime)$ is the pullback of the metric on 
$(K,\underline{z})$.}
\end{itemize}
Let $\hat{s}(\sigma^\prime)\in K$ denote the smallest simplex in 
which $s(\sigma^\prime)$ is contained. Obviously, if 
$\sigma^\prime$ is a $k^\prime$ simplex, then 
$\hat{\sigma}^\prime$ is $k$ simplex with $k\ge k^\prime$.
We shall write $\sigma^\prime\preceq \sigma$ if 
$\sigma=\hat{s}(\sigma^\prime)$ and 
$\sigma^\prime\npreceq\sigma$ otherwise.

We shall almost exclusively consider special simplicial complexes, which are given as follows.\\
An $n-$dimensional {\it pseudomanifold} is a finite simplicial complex $K^n$ such that 
\begin{enumerate}
 \item{Every simplex is a face of some $n$-simplex.}
 \item{Every $(n-1)$-simplex is the face of at most two $n$-simplexes.}
 \item{If $\sigma$ and $\sigma^\prime$ are $n$-simplexes of $K^n$, there is a finite sequence 
$\sigma=\sigma_1,
\cdots, \sigma_m=\sigma^\prime$ of $n$-simplexes of $K$, such that $\sigma_i$ 
and $\sigma_{i+1}$ have an $(n-1)$-simplex in common.}
\end{enumerate}
Unless otherwise stated, the dimension $n$ will always be taken to be $\ge 3$.
The (possibly empty) boundary $\partial K^n$ of $K^n$ is the subcomplex formed by 
the $(n-1)$-simplexes, which lie in exactly one $n$-simplex, and their faces. 
The third condition guarantees that $^sK^n$ is connected. $\partial K^n$ is not necessarily a 
pseudomanifold.

As an example, $\sigma^n$ is an $n$-dimensional pseudomanifold and its boundary $\partial \sigma^n
=\Sigma^{n-1}(\sigma^n)$ is an $(n-1)$-dimensional pseudomanifold without boundary. 

A pseudomanifold $K^n$ is called {\it orientable} if and only if 
$H_{n}(K^n,\partial K^n)\simeq\Z$ and \newline $H_{n-1}(K^n,\partial K^n)$ has no torsion. 
An {\it orientation} is a choice of a generator of $H_{n}(K^n,\partial K^n)$.
The {\it volume} of the p.l. space  $(K^n,\underline{z})$ is defined to be
\begin{equation*}
 V(K^n,\underline{z})=\sum_{\sigma^n\in K^n}|\sigma^n|(\underline{z})>0.
\end{equation*}
The scaling law 
\begin{equation}\label{volumescaling}
 V(K^n,\lambda\underline{z})=\lambda^{n/2} V(K^n,\underline{z})
\end{equation}
is clear. It compares with the scaling law for the volume in Riemannian geometry, see \eqref{g:13}.
A {\it smooth triangulation} of an $n$-dimensional smooth manifold $M$ is a pair $(K,\phi)$, 
where $K$ a simplicial complex and $\phi$ a 
homeomorphism from $^sK$ onto $M$ such that its restriction $\phi|\bar{\sigma}$ to any 
closed simplex $\bar{\sigma}\subset ^sK$ is smooth.
A well known theorem says that any compact smooth and connected manifold $M$ has a
smooth triangulation with finite $K$, which actually is a pseudomanifold (see e.g. \cite{Mu}).

For $\sigma^{n-2}\subset\sigma^n\in K$ let the unique 
$\sigma^{n-1}_1, \sigma^{n-1}_2\subset \sigma^n$ 
be such that $\sigma^{n-2}=\sigma^{n-1}_1\cap\sigma^{n-1}_2$. In their realization as 
euclidean simplexes in $E^n$, let $n_1$ and $n_2$ be unit vectors, 
normal to $\sigma^{n-1}_1$ and $\sigma^{n-1}_2$ respectively and pointing outwards. 
Then the {\it dihedral angle} $0<(\sigma^{n-2},\sigma^n)<1/2$ (in units of $2\pi$) is defined as 
\begin{equation*}
(\sigma^{n-2},\sigma^n)=\frac{1}{2}-
\frac{1}{2\pi}\arccos\langle n_1, n_2\rangle.
\end{equation*} 
$\langle\cdot,\cdot\rangle$ denotes the euclidean scalar product.
The two limiting (and degenerate) cases are $n_1=-n_2$, for which the dihedral angle vanishes, 
and $n_1=n_2$, 
for which the dihedral angle equals $1/2$.
In Appendix \ref{app:2} we shall provide another description of the dihedral angle.

The following scale invariance is obvious
\begin{equation}\label{scaling2}
 (\sigma^{n-2},\sigma^n)(\lambda \underline{z})=
(\sigma^{n-2},\sigma^n)(\underline{z}),\qquad \lambda>0.
\end{equation}
To a given p.l. space $(K^n,\underline{z})$, with $K^n$ being an $n$-dimensional pseudomanifold, 
we associate its {\it total scalar curvature}
\begin{equation}\label{totalscalar1}
 \cR(K^n,\underline{z})=\sum_{\sigma^{n-2}}\cR_{\sigma^{n-2}}(K^n,\underline{z})=\sum_{\sigma^{n-2}}
\left(1-\sum_{\sigma^n\,:\,\sigma^{n}\supset \sigma^{n-2}}(\sigma^{n-2},\sigma^n)\right)
|\sigma^{n-2}|(\underline{z})
\end{equation}
and the {\it average scalar curvature}
\begin{equation}\label{avtotalscalar}
 \overline{\cR}(K^n,\underline{z})=\frac{\cR(K^n,\underline{z})}{V(K^n,\underline{z})}.
\end{equation}

The expression in braces in \eqref{totalscalar1} is called the {\it deficit angle at $\sigma^{n-2}$} and will be 
written as $\delta(\sigma^{n-2})=\delta(\sigma^{n-2})(K^n,\underline{z})$.
When $K$ is not a pseudomanifold, the definition is slightly different, 
see \cite{CMS2}, where also 
p.l. versions of {\it Lipschitz-Killing curvatures} are given.
The total scalar curvature does not change under a subdivision (and the same is valid for the volume), that is
\begin{equation}\label{totalscalar10}
 \cR(K^{n\,\prime},\underline{z}^\prime)=\cR(K^n,\underline{z})
\end{equation}
holds whenever $(K^{n\,\prime},\underline{z}^\prime)$ is a subdivision of $(K^n,\underline{z})$.
 For further use let us briefly see how this comes about. 
 First the additivity of volumes gives
\begin{equation*}
 \sum_{\sigma^{n-2\; \prime}\;:\;\sigma^{n-2\; \prime}\;\preceq\; \sigma^{n-2}}|\sigma^{n-2\; \prime}|=|\sigma^{n-2}|
\end{equation*}
for all $\sigma^{n-2}$. Also for any pair $\sigma^{n-2}\subset \sigma^n$ the following relation 
holds between deficit angles
\begin{equation*}
\delta(\sigma^{n-2\;\prime})=\delta(\sigma^{n-2})
\end{equation*}
for all $\sigma^{n-2\;\prime}\;\preceq\;\sigma^{n-2}$.
These two relations prove that 
\begin{equation*}
 \cR_{\sigma^{n-2}}(K,\underline{z})=
 \sum_{\sigma^{n-2\;\prime}\;:\;\sigma^{n-2\;\prime}\preceq\sigma^{n-2}}
 \cR_{\sigma^{n-2\;\prime}}(K^\prime,\underline{z}^\prime)
\end{equation*}
holds for all $\sigma^{n-2}$. 
Set 
\begin{equation*}
\Theta^{n-2}=
\{\sigma^{n-2\;\prime}\;|
\dim \sigma^{n-2\;\prime}<\dim \hat{s}(\sigma^{n-2\;\prime})\}.
\end{equation*}
Then 
\begin{equation*}
\delta(\sigma^{n-2\;\prime})=0,\qquad
 \sigma^{n-2\;\prime}\in\Theta^{n-2},
\end{equation*}
 in other words 
$(K^\prime,\underline{z}^\prime)$ is flat around 
$\sigma^{n-2\;\prime}\in\Theta^{n-2}$.
Therefore
$$
\cR_{\sigma^{n-2\;\prime}}(K^\prime,\underline{z}^\prime)
=0,\qquad \sigma^{n-2\;\prime}\in\Theta^{n-2}.
$$
by \eqref{totalscalar1}. This establishes \eqref{totalscalar10}.

\eqref{scaling1} and \eqref{scaling2} give
\begin{equation}\label{scaling3}
 \cR(K^n,\lambda \underline{z})=\lambda^{(n-2)/2}\cR(K^n,\underline{z}),
\end{equation}
which compares with the scaling behavior of the total scalar curvature in Riemannian geometry, see again \eqref{g:13}.
We call the gradient of the total scalar curvature
\begin{equation}\label{ricci}
 \underline{Ein}(K^n,\underline{z})=\Big\{Ein_{\sigma^1}(K^n,\underline{z})\Big\}_{\sigma^1\in K^n}
 =\underline{\nabla} \cR(K^n,\underline{z})=
\Big\{\frac{\partial}{\partial z_{\sigma^1}}\cR(K^n,\underline{z})\Big\}_{\sigma^1\in K^n}
\end{equation}
the {\it Einstein vector field}. By definition $(K^n,\underline{z})$ is { \it Einstein flat at $\sigma^1$} if 
$Ein_{\sigma^1}(K^n,\underline{z})=0$ and {\it Einstein flat if $\underline{Ein}(K^n,\underline{z})=0$}.  
\begin{remark}\label{re:ric}
Here is the time to point out an important difference between curvature concepts of p.l.\ geometry and those in 
Riemannian geometry.
Despite many efforts, so far no tensor calculus has been formulated. In particular no analogues of the metric tensor, of the curvature, 
or the Riemannian curvature tensor, 
or the Ricci tensor - all pointwise defined quantities on the underlying manifold- have been found. So the main analogies 
may be found between globally defined objects, like the volume or the total scalar curvature. 
Now the Ricci tensor or rather the Einstein tensor shows up in the variation of 
the total curvature, see \eqref{g:16}, so by comparison with \eqref{ricci}, this is the closest we can get to the 
Ricci tensor by analogy in the theory of p.l.\ spaces, thus our choice of notation.
\end{remark}
Since we will make intensive use of Euler's relation, we briefly recall it within the present context. 
Also Appendix \ref{app:Einstein} provides the corresponding formulation in Riemannian geometry. 
Let 
$$
\langle\underline{z},\underline{z}^\prime\rangle=\sum_{\sigma^1}z_{\sigma^1}z^\prime_{\sigma^1}
$$
denote the euclidean scalar product and $||\underline{z}||^2=\langle\underline{z},\underline{z}\rangle$.
Observe that
$$||\underline{z}||^2=
\sum_{\sigma^1\in K^n}l_{\sigma^1}^4
$$
is the fourth moment of the edge lengths.

By definition any (smooth) function $f(\underline{z})$ is homogeneous of order $m$ if 
$f(\lambda\underline{z})=\lambda^m f(\underline{z})$ is valid for all metrics $\underline{z}$.
\begin{lemma}\label{lem:euler} {\rm (Euler's Relation)}
If $f(\underline{z})$ is homogeneous of order $m$ and differentiable, then
\begin{equation*}
 \langle\underline{z},\underline{\nabla}f(\underline{z})\rangle =mf(\underline{z}).
\end{equation*}
holds. In particular if $f$ is of homogeneous of order $m \neq 0$ and if $\underline{z}_{crit}$ is a critical point of $f$ - 
such that actually 
all points $\lambda \underline{z}_{crit}$ are critical - then $f(\underline{z}_{crit})=0$.
\end{lemma}
As will be seen below, this lemma turns out to be a surprisingly efficient tool for the present context .
A consequence of \eqref{scaling3} is
\begin{equation}\label{scaling4}
\frac{n-2}{2}\cR(K^n,\underline{z})=\langle\underline{z},\underline{Ein}(K^n,\underline{z})\rangle
\end{equation}
is valid for all metrics $\underline{z}$. Here we have used the differentiability w.r.t. 
$\underline{z}$. This property will become clear from the discussion to be given below.
 Here and in what follows, we view $\underline{z}$ as the tautological vector field.
Observe that $\underline{z}$, like $\underline{Ein}(K,\underline{z})$, is a gradient due to 
\begin{equation*}
\underline{z}=\underline{\nabla}\frac{1}{2}||\underline{z}||^2.
\end{equation*}
A more explicit expression for the Einstein vector field is
\begin{equation}\label{totalscalar3}
Ein_{\sigma^1}(K^n,\underline{z})=
\sum_{\sigma^{n-2}}
\left(1-\sum_{\sigma^{n}\supset\sigma^{n-2}}(\sigma^{n-2},\sigma^n)\right)
\partial^{\sigma^1}|\sigma^{n-2}|.
\end{equation}

The proof is obtained by using the Leibniz rule and a remarkable  
formula of Regge \cite{Regge}, by which 
\begin{equation}\label{Reggevar}
\sum_{\sigma^{n-2}\,:\,\sigma^{n-2}\subset \sigma^n }
(\sigma^{n-2},\sigma^n)^\prime|\sigma^{n-2}|=0\quad \mbox{for all}\quad 
\sigma^{n}
\end{equation}
holds for any variation of $\underline{z}$, and 
where now $^\prime$ denotes the derivative with respect to the variation. For another proof 
see also \cite{CMS2}. 
\begin{remark}\label{re:history}
Due to the importance of \eqref{Reggevar} for our central relation \eqref{totalscalar3} a historic remark at this place might be 
appropriate. \eqref{Reggevar} is often mentioned in connection with {\rm Schl\"afli's formula}, which is a variation formula for 
Euclidean and non-Euclidean volumes. In 1858 Schl\"afli provided such a relation for the volume of spherical simplexes 
\cite{Schlaefli}. In 1907 {\rm Sforza} extended this to the case of simplexes in {\rm Lobachevsky space} \cite{Sforza}. 
Modern proofs of these results may be found in \cite{Boehm,BoehmHertel,Kneser}. The extension to polyhedra is easy.
It was {\rm Milnor} who provided a unified formula,
which includes Euclidean polyhedra as well and which reads as follows \cite{Milnor}
\begin{equation*}*\label{milnor}
 K|P^n|^\prime=\frac{1}{n-1}\sum _{P^{n-2}\subset \partial P^n}(P^{n-2},P^{n})^\prime|P^{n-2}|.
\end{equation*}
The notation is the following. $P^n$ is a polyhedron in $M^n$, that is a finite intersection of half spaces and which is compact. 
$M^n$ itself is a space of constant sectional 
curvature $K$. $P^{n-2}$ is an $(n-2)-$ dimensional face of $P^n$. $(P^{n-2},P^{n})$ is the dihedral angle in analogy to 
$(\sigma^{n-2},\sigma^n)$ and $|P^n|$ and $|P^{n-2}|$ are their $n$- and $(n-2)$-dimensional volumes 
respectively. 
As is visible from \eqref{milnor}, {\rm Milnor} put particular emphasis on the transition between Euclidean and non-Euclidean cases 
for $K$ near zero. Observe that the cases of arbitrary $K\neq 0$ can be obtained from the cases $K=\pm 1$ by appropriate scaling.
For simplexes and the choice $K=0$ \eqref{milnor} is just Regge's relation \eqref{Reggevar}.
\end{remark}
The scaling behavior 
\begin{equation}\label{Ricscale}
 \underline{Ein}(K^n,\lambda\underline{z})=\lambda^{(n-4)/2}\underline{Ein}(K^n,\underline{z})
\end{equation}
is obvious. The scaling relations \eqref{volumescaling} and \eqref{scaling3} fit 
with the corresponding scaling 
relations \eqref{g:13} in the smooth case.

It is tempting to call 
\begin{equation}\label{totalscalar4}
 Sec_{\sigma^{n-2}}=Sec_{\sigma^{n-2}}(K^n,\underline{z})=
\left(1-\sum_{\sigma^{n}\supset\sigma^{n-2}}(\sigma^{n-2},\sigma^n)\right)
\end{equation}
the {\it sectional curvature} at the 2-plane orthogonal to $\sigma^{n-2}$. 
Note, however, that it is scale invariant in contrast to the 
sectional curvature in Riemannian geometry.

From \eqref{scaling4} we immediately obtain the following result.
We say that $\underline{v}=\{v_{\sigma^1}\}_{\sigma^1\in K},v_{\sigma^1}\in \R $ is 
non-negative or non-positive, if $v_{\sigma^1}$ is non-negative or non-positive for all $\sigma^1$.
$\underline{v}$ is strictly positive or strictly negative, if every component $v_{\sigma^1}$ is 
positive or negative respectively. Any metric $\underline{z}$ is strictly positive.
\eqref{scaling4} then directly gives
\begin{proposition}\label{prop:0}
If $\underline{Ein}(K^n,\underline{z})\;(n\ge 3)$ is non-negative or non-positive, then 
the total scalar  curvature is also non-negative or non-positive respectively.
If $\underline{Ein}(K^n,\underline{z})$ is strictly positive or strictly negative, 
then the total scalar curvature is also positive or negative respectively.
\end{proposition}
Observe that both the sectional curvature \eqref{totalscalar4} and 
$$
\partial^{\sigma^1}|\sigma^{n-2}|,\qquad 
\partial^{\sigma^1}|\sigma^n|
$$
may become positive or negative.
There is another vector field, which is also a gradient field, namely the gradient of the volume
\begin{equation*}
\underline{v}(K^n,\underline{z})=\underline{v}(\underline{z})=\left\{v_{\sigma^1}(\underline{z})\right\}_{\sigma^1\in K^n}
=\underline{\nabla}V(K^n,\underline{z})
\end{equation*}
with the scaling behavior
\begin{equation}\label{volgrad2}
 \underline{v}(K^n,\lambda\underline{z})=\lambda^{(n-2)/2}\underline{v}(K^n,\underline{z}),\qquad \lambda >0.
\end{equation}
By \eqref{volumescaling} and Euler's relation
\begin{equation}\label{volgrad3}
\langle\underline{z},\underline{v}(K^n,\underline{z})\rangle =\frac{n}{2}V(K^n,\underline{z}) 
\end{equation}
holds.
\begin{remark}\label{remark:2}
Instead of using $\underline{z}$, the set of squares of the edge lengths, one could as well use the set 
$$
\underline{l}=\left\{l_{\sigma^1}\right\}_{\sigma^1\in \cK}
$$
of edge lengths themselves to parametrize a euclidean metric. All definitions easily carry over. Thus one 
might consider the gradient 
$\underline{\nabla}\,_{\underline{l}}\cR(K^n,\underline{l})$ 
of the total scalar curvature w.r.t. $\underline{l}$. For $n=3$ the sectional curvature \eqref{totalscalar4} and the 
Einstein vector field agree. The definitions of 
Einstein metrics and (normalized) Einstein flows to be given below, can also be adapted to this choice of 
parametrization.
As we shall observe below, see Remark \ref{remark:3}, the definitions for Einstein metric are not equivalent.
In Section \ref{sec:allmetrics} we shall argue, that it is more appropriate to use $\underline{z}$ to describe 
the space of all 
metrics on a given simplicial complex $K$.
\end{remark}


\section{The space of all metrics}\label{sec:allmetrics}

In this section we will establish some properties of the set of all metrics on a given 
finite $n$-dimensional pseudomanifold $K^n$. 
In particular we will show, as announced, that the squares 
$z_{\sigma^1}=l^2_{\sigma^1}$ of the edge lengths are better suited to parametrize the set of all metrics
and we will use the notation $\partial^{\sigma^1}$ for the partial derivative w.r.t.\ the variable $z_{\sigma^1}$.
Let $n_1(K^n)$ denote the number of $1$-simplexes in $K^n$. With this convention the set $\cM(K^n)$ of 
all metrics on $K^n$ can be viewed as a subset of $\R_+^{n_1(K^n)}$. This set is non-empty, the choice 
where all $l^2_{\sigma^1}$ are equal serves as an example. As a matter of fact $^sK^n$ itself carries 
the metric, for which $l^2_{\sigma^1}=2$ for all $\sigma^1$.
The relation $n_1(\Sigma^k(K^n))=n_1(K^n)$ for all $1\le k \le n$ is obvious, as is 
$\Sigma^j(\Sigma^k(K^n))=\Sigma^j(K^n)$ for $1\le j\le k \le n$. By definition of $\cM(K^n)$, the set
$\cM(K^n)$ can be viewed as a subset of $\cM(\Sigma^k(K^n))$ for every $k$. It is easy to verify that 
it always is a proper subset for $k<n$, that is $\cM(\Sigma^k(K^n))\neq \cM(K^n)$.

Thus we have the chain
\begin{equation*}
\cM(K^n)\subset\cM(\Sigma^{n-1}(K^n))\subset \cM(\Sigma^{n-2}(K^n))\subset\cdots\subset\cM(\Sigma^{1}(K^n))
=\R_+^{n_1(K^n)}.
\end{equation*}
The main result of this section is the 
\begin{theorem}\label{theo:5}
 $\cM(K^n)$ is an open convex cone in $\R_+^{n_1(K^n)}$. In particular $\cM(K^n)$ is connected.
\end{theorem}
We note another analogy with the smooth case. Indeed, the set of all Riemannian metrics on a manifold forms a convex cone in 
the set of all second order tensor fields.
\begin{proof}
First consider a euclidean $k$-simplex $\sigma^k$ in $E^k$ and label its vertices in an arbitrary order as 
$0,1,\cdots, k$. Assume the vertex $0$ is placed at the origin. Again we regard the other vertices as being
represented by the (linearly independent) vectors $v_i, 1\le i\le k$.
Then the length $l_{ij}=l_{ji}$ of the edge connecting the two different vertices $i$ and $j$ 
is given in the form 
\begin{align*}
z_{0j}&=l_{0j}^2=\langle v_j,v_j\rangle,\qquad\qquad \qquad 1\le j\le k,\\\nonumber
z_{ij}&=l_{ij}^2=\langle v_i- v_j,v_i-v_j\rangle,\qquad 1\le i,j\le k.
\end{align*}
We make the convention $z_{ii}=l_{ii}^2=0$.
As a consequence the $k\times k$ real, symmetric matrix $A=A(\underline{z})$ , 
$\underline{z}=\{z_{ij}\}_{0\le i,j\le k}$, with entries
\begin{equation}\label{amatrix}
 a_{ij}=a_{ji}=\langle v_i,v_j\rangle=\frac{1}{2}(z_{0i}+z_{0j}-z_{ij}),\qquad 1\le i,j\le k
\end{equation}
is positive definite. The volume of the euclidean $k$-simplex is then obtained as
\begin{equation}\label{amatrixvol}
|\sigma^k|=|\sigma^k|(\underline{z})=\frac{1}{k!}\det A^{1/2}=
\frac{1}{k!}
\left(\langle v_1\wedge v_2\wedge \cdots v_k,
v_1\wedge v_2\wedge \cdots v_k\rangle\right)^{1/2}.
\end{equation}
For the particular case $k=2$ this relation gives the area of a triangle in terms 
of its edge lengths (squared), originally attributed to Heron of Alexandria. 
The following lemma is trivial.
\begin{lemma}
The following estimate is valid for any pair $\sigma^l\subset\sigma^k$ in $K^n$. 
There are universal constants 
$c_n$, such that 
\begin{equation*}
 |\sigma^k(\underline{z})|\le c_n ||\underline{z}||^{(k-l)/2}|\sigma^l(\underline{z})|
\end{equation*}
holds for any metric $\underline{z}$.
\end{lemma}
It is also clear that in general any volume $|\sigma^k(\underline{z})|$ will not stay away from zero even 
if $||\underline{z}||$ stays away from zero. For the same reason
\begin{equation}\label{dervol}
\partial^{\sigma^1}|\sigma^k(\underline{z})|
=\frac{1}{2|\sigma^k(\underline{z})|}\partial^{\sigma^1}\det A(\underline{z}(\sigma^k))
\end{equation}
for $\sigma^1\in\sigma^k$ may become unbounded even if $||\underline{z}||$ stays away from zero.
With the above notation we have the
\begin{lemma}\label{sympol}
$\det A(\underline{z})$ is a homogeneous, symmetric polynomial of order $k$ in the $z_{ij},0\le i<j\le k$.
\end{lemma}
This lemma shows that the above result \eqref{amatrixvol} is independent of the particular labeling of the 
vertices in $\sigma^k$.
For the case $n=2$, see the Example \ref{triangle} below.
\begin{proof}
Homogeneity and the order are clear. Symmetry follows from a geometric argument. The construction above was based 
on a particular choice of the 
order of labeling. We could have as well chosen an arbitrary other order, which amounts to a permutation of 
the $k+1$ vertices. This would result 
in another construction of the euclidean $k$-simplex with the same volume. The claim then follows from 
\eqref{amatrixvol}.
\end{proof}
The converse is also valid. For any real positive definite $k\times k$ 
matrix $A$, invert \eqref{amatrix} to {\it define} lengths squares as 
\begin{align}\label{Atol}
z_{0i}&=l_{0i}^2=a_{ii},\\\nonumber
z_{ij}&=l_{ij}^2=a_{ii}+a_{jj}-2a_{ij}.
\end{align}
Since $A$ is positive definite, one can build a euclidean $k$-simplex with these edge 
lengths (squared). 
The following lemma is well known, see e.g.  \cite{Bhatia,TerrasII}. 
Via the above correspondence it provides a higher dimensional extension of the triangle 
inequality for the three edge lengths of a euclidean triangle, see Example \ref{triangle} below. 
\begin{lemma}\label{tartaglia}
Let any symmetric $k\times k$ matrix $B$ be given with entries labeled by the set $\{1,\cdots, k\}$. 
Set $I_l=\{1,2,\cdots ,l\}$ with $1\le l\le k$ and let $B_l$ denote the $l\times l$ matrix 
$\{B_{ij}\}_{i,j\in I_l}$. Then $B$ is positive definite if and only if $\det B_k >0$ 
holds for all $k$.
\end{lemma}
Since $A(\lambda \underline{z})=\lambda A(\underline{z})$, we conclude that the set of $\underline{z}$,
for which one can build a euclidean $k$-simplex with these edge lengths (squared), 
is an open cone in $\R_+^{k(k+1)/2}$.
Moreover, this cone is convex. Indeed, by definition of $A(\underline{z})$ the relation 
\begin{equation*}
 A(\rho \underline{z}+(1-\rho)\underline{z}^{\prime})=
\rho A(\underline{z})+(1-\rho)A(\underline{z}^{\prime}),\qquad 0\le \rho\le 1
\end{equation*}
is obvious.
The claim now follows directly from the fact, that a convex combination of two positive 
definite matrices is again positive definite. By \eqref{Atol} the corresponding edge lengths 
squares are of the form
\begin{equation}\label{edgelength}
\rho\underline{z}+(1-\rho)\underline{z}^{\prime}=\underline{z}^{\prime\prime}=
\{z^{\prime\prime}_{ij}=
\rho z_{ij}+(1-\rho)z^{\prime}_{ij}\}_{i\le j}.
\end{equation}
\begin{example}\label{triangle}
 For $k=2$
$$
A(\underline{z})=\left(\begin{matrix}z_{01}&\frac{1}{2}\left(z_{12}-z_{01}-z_{02}\right)\\
                       \frac{1}{2}\left(z_{12}-z_{01}-z_{02}\right)&z_{02} 
                       \end{matrix}\right)
$$
such that
$$
\det A(\underline{z})=\frac{1}{2}\left(z_{01}z_{02}+z_{01}z_{12}+z_{02}z_{12} \right)-\frac{1}{4}\left(z_{01}^2+z_{02}^2+z_{12}^2\right).
$$
Therefore the two conditions $z_{01}>0$ and $\det A(\underline{z})>0$ are equivalent to the three conditions
$ \sqrt{z_{01}}<\sqrt{z_{02}}+\sqrt{z_{12}},\sqrt{z_{02}}<\sqrt{z_{01}}+\sqrt{z_{12}}$ and 
$\sqrt{z_{12}}<\sqrt{z_{01}}+\sqrt{z_{02}}$.
In particular the first two  conditions imply $z_{02}>0$ and $z_{12}>0$.
\end{example}

This discussion for a single simplex $\sigma^k$ carries over to all simplexes in $K$ as follows. Indeed, to 
see that $\cM$ is convex, consider now the 
convex combination \eqref{edgelength} now with $\underline{z},\underline{z}^\prime\in\cM$. For any 
$\sigma^k\in K$ set 
\begin{equation*}
\underline{z}(\sigma^k)=\{z_{\sigma^1}\}_{\sigma^1\subseteq\sigma^k}.
\end{equation*}
With the notation of \eqref{edgelength} and by the discussion above 
$$
\underline{z}^{\prime\prime}(\sigma^k)=(\rho\underline{z}+(1-\rho)\underline{z}^{\prime})(\sigma^k)
=\rho\underline{z}(\sigma^k)+(1-\rho)\underline{z}^{\prime})(\sigma^k)
$$
it follows that one can build a euclidean $k$-simplex with edge lengths squared equal to 
$\underline{z}^{\prime\prime}(\sigma^k)$. 
Since this holds for 
all $\sigma^k\in K^n$, this establishes that $\cM(K^n)$ is convex. This result is the main reason for having 
chosen
the squares of the edge lengths as the basic parameters for a metric.
Moreover $\underline{z}\in\cM(K^n)$ implies $\lambda\underline{z}\in\cM(K^n)$ for any $\lambda>0$, so 
$\cM(K^n)$ is a convex cone. With the choice of the $l_{\sigma^1}$ as parameters convexity would fail.
For any $\sigma^k\in K^n$, let $\underline{z}(\sigma^k)$ denote the set of 
$z_{\sigma^1}$ 
with $\sigma^1\in\sigma^k$. Define $A(\underline{z}(\sigma^k))$ according to the procedure given above.
Then $\underline{z}\in\cM(K^n)$ if and only if $\det A(\underline{z}(\sigma^k))>0$ for all $\sigma^k\in K^n$.
Since each map $\underline{z}\mapsto \det A(\underline{z}(\sigma^k))$ is continuous, this proves that $\cM(K^n)$ is
open.
\end{proof}

Actually the set $\cM(K)$ is a Riemannian manifold 
in a canonical way. We first consider a single 
$n$-simplex. Let $\cP_n$ denote the space
of all real, positive definite $n\times n$ matrices. This space is a Riemannian manifold of dimension 
$n(n+1)/2$, see e.g. \cite{Bhatia,TerrasII}. The pullback of the metric on $\cP_n$ to $\cM(\sigma^n)$ via the 
one-to-one smooth map $\phi\,: \cM(\sigma^n)\;\rightarrow\;\cP_n$ given by \eqref{amatrix}
turns $\cM(\sigma^n)$ into a Riemannian manifold. Now consider the Riemannian manifold 
$$
\times_{\sigma^n\in K^n}\cM(\sigma^n).
$$ 
Write a point in this space as $\times_{\sigma^n\in K^n}\underline{z}(\sigma^n)$. 

$\cM(K)$ is now obtained as a closed submanifold of this space. Indeed, consider any metric 
$\underline{z}$ on $K$ and any edge $\sigma^1\in K$, which is the face of any
$\sigma^n$ and $\sigma^{n\,\prime}$. Then its edge length squared $z_{\sigma^1}$ defines a metric on 
both $\sigma^n$ and $\sigma^{n\,\prime}$. With the above notation this is just the condition 
$z_{\sigma^1}(\sigma^n)-z_{\sigma^1}(\sigma^{n\,\prime})=0$. Going through all such triples in 
$K$ the collection of all these conditions define $\cM(K)$.
By this discussion we also see that $\cM(K^n)$ is given as
$$
\cM(K^n)=\left\{\underline{z}\in \R^{n_1(K^n)}\;|\; 
\det A(\underline{z}(\sigma^k))>0\;\mbox{for all}\;\sigma^k\in K^n\right\}.
$$
We now introduce a quantity, which serves to measure the distance of a metric 
$\underline{z}\in \cM(K^n)$ to the boundary $\partial\cM(K^n)$ of $\cM(K^n)$, defined as 
$\partial\cM(K^n)=\overline{\cM(K^n)}\setminus \cM(K^n)$. $\overline{\cM(K^n)}$ denotes the closure of $\cM(K^n)$. 
Indeed, set 
\begin{equation*}
d(\underline{z}) =
\min_{k}\min_{\sigma^k\in K}\left(k!\sqrt{\det A(\underline{z}(\sigma^k))}\right)^{2/k}
=\min_{k}\min_{\sigma^k\in K}|\sigma^k|(\underline{z})^{2/k}.
\end{equation*}
This quantity has the right scaling behavior: 
\begin{equation*}
 d(\lambda\underline{z})=\lambda d(\underline{z}), \quad \lambda >0.
\end{equation*}


\section{Einstein metrics}\label{sec:einsmet}
The existence of the two vector fields $\underline{z}$ and $\underline{v}$ leads us to two alternative 
and hence different definitions of Einstein metrics.
\begin{definition}\label{def:einstein}
For given pseudomanifold $K^n$ a metric $\underline{z}_0$ 
is an {\it Einstein metric on $K^n$} of {\rm type I}, if there is real constant $\kappa_I$, such that
\begin{equation}\label{ricci1}
\underline{Ein}(K^n,\underline{z}_0)-\kappa_I\, \underline{z}_0=0\\
\end{equation}
holds. 

For given pseudomanifold $K^n$ a metric $\underline{z}_0$ 
is an {\it Einstein metric on $K^n$} of {\rm type II} if there is real constant $\kappa_{II}$, such that
\begin{equation}\label{vol1}
\underline{Ein}(K^n,\underline{z}_0)-\kappa_{II}\, \underline{v}(K^n,\underline{z}_0)=0
\end{equation}
holds. In both cases $(K^n,\underline{z}_0)$ is then called a {\it piecewise linear (p.l.) Einstein space}.
\end{definition}
A p.l.\ space $(K^n,\underline{z})$ is said to be {\it Einstein-flat at $\sigma^1$} if 
$Ein_{\sigma^1}(\underline{z})=0$. 
$(K^n,\underline{z})$ is said to be {\it Einstein-flat}, if it is Einstein-flat at all $\sigma^1$. 
An Einstein-flat p.l. space is also a p.l. Einstein space of both types with vanishing 
$\kappa_I$ and $\kappa_{II}$. Also it has vanishing total scalar curvature and therefore also vanishing mean scalar curvature.

Equation \eqref{vol1} is a p.l. analogue of the Einstein vacuum equations with a cosmological term, that is
$\kappa_{II}$ plays the r\^{o}le of a { \it cosmological constant}.
The condition \eqref{ricci1} for an Einstein metric $\underline{z}_0$ of type I is local in the following sense. 
Component wise it reads
\begin{equation}\label{ricci2}
Ein_{\sigma^1}(K^n,\underline{z}_0)-\kappa_I z_{0\,\sigma^1}=0.
\end{equation}
For any $k$-simplex $\sigma^k\in K^n$, $0\le k\le n-1$ 
its star, denoted by $\mbox{\it star}(\sigma^k)$, is the subcomplex of $K$ consisting of all 
$\sigma^n\supset \sigma^k$ and its 
faces. Then in \eqref{ricci2} the l.h.s. is only a function of those 
$z_{\sigma^{1\,\prime}}$ for which $\sigma^{1\,\prime}\in\mbox{\it star}(\sigma^1)$. 

Similarly the condition \eqref{vol1} for 
an Einstein metric $\underline{z}_0$ of type II is also local.
These definitions mimic the standard definition of an Einstein space in the smooth case, see \eqref{g:8}.
\begin{remark}\label{remark:3}
With the notation as in Remark \ref{remark:2}, if one replaces the definition \label{riici1} for an Einstein metric of 
type I by the condition
\begin{equation}\label{riccil1}
\underline{\nabla}\,_{\underline{l}}\cR(K^n,\underline{l})-\kappa_I^\prime\, \underline{l}=0\\
\end{equation}
then {\rm a priori} these two conditions do not give rise to the same solutions. This is easily seen using the trivial 
identity
$$
\frac{\partial}{\partial l_{\sigma^1}}=2l_{\sigma^1}\frac{\partial}{\partial z_{\sigma^1}}.
$$
A corresponding statement 
holds if condition \eqref{vol1} is replaced by the condition
\begin{equation}\label{voll1}
\underline{\nabla}_{\underline{l}}\cR(K^n,\underline{l})-\kappa_{II}^\prime\,\underline{\nabla}_{\underline{l}}V(K^n,\underline{l})=0.
\end{equation}
The metrics provided in Examples \ref{ex:0} and \ref{ex:1} below, the only known so far, satisfy all four conditions \eqref{ricci1}, 
\eqref{vol1}, \eqref{riccil1} and \eqref{voll1}.
For the special case $n=3$, $\underline{\nabla}\,_{\underline{l}}\cR(K^3,\underline{l})$ is just the sectional 
curvature \eqref{totalscalar4}. This quantity was then used in \cite{CGY} to give two definitions of 
an Einstein metric in analogy to definition of the two types I and II just given.
\end{remark}
The following proposition is an immediate consequence of \eqref{Ricscale} and \eqref{volgrad2} and holds for both types of 
Einstein metrics.
\begin{proposition}\label{prop:11}
 If $(K^n,\underline{z})$ is a p.l. Einstein space, so is $(K^n,\lambda\underline{z})$  for any $\lambda >0$.
\end{proposition}

\vspace{0.5cm}
\subsection{Einstein metrics of type I}~~\\
In this subsection $K^n$ ( and therefore in particular $n\ge 3$) will be fixed, so from now on, 
and when the context is clear, we will simply write 
$\cR(\underline{z}),\underline{Ein}(\underline{z}), V(\underline{z})$ etc.

\begin{proposition}\label{prop:1}
Let $(K^n,\underline{z}_0),\, (n\ge 3)$ be a p.l.\ Einstein space of type I. Then $\cR(\underline{z}_0)$ 
is strictly positive or strictly negative if and only if $\underline{Ein}(\underline{z}_0)$ is strictly positive or 
strictly negative respectively. Also $\cR(\underline{z}_0)$ vanishes if and only if $(K^n,\underline{z}_0)$ is Einstein-flat.
\end{proposition}
The last statement is also valid for a p.l.\ Einstein space of type II. This result for Einstein spaces of 
type I extends the result in Proposition \ref{prop:0}.
\begin{proof}
The first claim follows trivially from the defining relation \eqref{ricci1}. If $(K^n,\underline{z}_0)$ is Einstein-flat, then 
$\cR(\underline{z}_0)=0$ due to \eqref{scaling4}. Conversely assume 
$\cR(\underline{z}_0)=0$.  By the definition of an Einstein metric of type I all $Ein_{\sigma^1}(\underline{z}_0)$ have 
the same sign unless they all vanish. Now observe that the $z_{0\;\sigma^1}$ are all positive.
Since $\langle\underline{z}_0, \underline{Ein}(\underline{z}_0)\rangle=0$, again by 
\eqref{scaling4}, this relation can therefore only hold if all $Ein_{\sigma^1}(\underline{z}_0)$ vanish, that is 
$(K^n,\underline{z}_0)$ is Einstein-flat.
Assume next that $\underline{z}_0$ is an Einstein metric of type II, which in addition is Einstein-flat. 
But then again by \eqref{scaling4} $\cR(\underline{z}_0)=0$. Conversely, if $\cR(\underline{z}_0)=0$, then by \eqref{scaling4}
$$
0=\langle
\underline{z}_0, \underline{Ein}(\underline{z}_0)\rangle
=\kappa_{II}\langle\underline{z}_0,\underline{v}(\underline{z}_0)\rangle=\kappa_{II}\frac{2}{n}V(\underline{z}_0),
$$
having used the definition \ref{vol1}. But this is only possible if $\kappa_{II}=0$, so $(K^n,\underline{z}_0)$ is 
Einstein-flat.
\end{proof}
\begin{proposition}\label{pro:EI1}
Let $\underline{z}_0$ be an Einstein metric of type I. Then $\kappa_I$ is given in terms of the 
total scalar curvature as $\kappa_I^{(1)}(\underline{z}_0)$ where
\begin{equation}\label{ricci3}
\kappa_I^{(1)}(\underline{z})=
\frac{n-2}{2}\frac{\cR(\underline{z})}{||\underline{z}||^2}
\end{equation}
is well defined for all $\underline{z}\in\cM(K^n)$. 

Alternatively $\kappa_I$ is also given in terms of the 
Einstein vector field as $\kappa_I^{(2)}(\underline{z}_0)$ where 
\begin{equation}\label{ricci30}
\kappa_I^{(2)}(\underline{z})=\frac{2}{n}
\frac{\langle \underline{v}(\underline{z}),\underline{Ein}(\underline{z})\rangle}{V(\underline{z})},
\end{equation}
which is also well defined for all $\underline{z}\in\cM(K^n)$.
Finally in case $(K^n,\underline{z}_0)$ is not Ricci flat, $\kappa_I$ is also given in terms of the 
Einstein vector field and the total scalar curvature as $\kappa_I^{(3)}(\underline{z}_0)$ where 
\begin{equation}\label{ricci300}
\kappa_I^{(3)}(\underline{z}) =\frac{n-2}{2}\frac{||\underline{Ein}(\underline{z})||^2}
{\cR(\underline{z})},
\end{equation}
which is well defined outside the zero set of 
$\langle\underline{v}(\underline{z}),\underline{Ein}(\underline{z})\rangle$.
\end{proposition}
The analogue in the smooth case is given in \eqref{g:9}.
\begin{proof}
Take the scalar product of \eqref{ricci1}, with $\underline{z}_0$, and then use \eqref{scaling4} in combination 
with Euler's relation, see Lemma \ref{lem:euler}. This proves the first claim. As for the second claim, now take 
the scalar product of \eqref{ricci1}, now with $\underline{v}(\underline{z}_0)$, and use \eqref{volgrad3}. 
Finally the third claim 
follows by taking the scalar product of  \eqref{ricci1} with $\underline{Ein}(\underline{z}_0)$.
\end{proof}
In order to analyze Einstein metrics in more detail, we need some preparations. Recall that we view a metric $\underline{z}$
as an element of the euclidean space $\R^{n_1(K^n)}$.  
Let $S^{\bar{n}}(r)$ with $\bar{n}=n_1(K^n)-1$ denote the sphere of radius $r>0$. We set 
\begin{equation*}
\cM_r(K^n)=\cM(K^n)\cap S^{\bar{n}}(r)=\{\underline{z}\in\cM(K^n)\,|\, ||\underline{z}||=r\}.
\end{equation*}
For any $0\neq \underline{x}\in\R^{\bar{n}}$ let $P(\underline{x})$ denote the orthogonal projection onto the 
line defined by $\underline{x}$.
More explicitly
\begin{equation}\label{pxx}
P(\underline{x})\underline{y}=\frac{\langle\underline{x},\underline{y}\rangle}{||\underline{x}||^2}\;\underline{x}.
\end{equation}
The scale invariance 
\begin{equation*}
P(\lambda\underline{x})= P(\underline{x})
\end{equation*}
for all $\lambda>0$ is obvious. 
$Q(\underline{x})=\1-P(\underline{x})$ is the orthogonal projection
onto the tangent space $T_{\underline{x}}S^{\bar{n}}(r=||\underline{x}||)$ 
to $S^{\bar{n}}(r=||\underline{x}||)$ at the point ${\underline{x}}$.
Set
\begin{equation}\label{ricchat}
 \underline{\widehat{Ric}}_I(\underline{z})=Q(\underline{z})\underline{Ein}(\underline{z})
=\underline{Ein}(\underline{z})-\frac{n-2}{2}\cR(\underline{z})\frac{1}{||\underline{z}||^2}\underline{z},
\end{equation}
which is defined for all $\underline{z}\in \cM(K^n)$. By this definition of 
$\underline{\widehat{Ric}}_I(\underline{z})$ and since
$Q(\underline{z})$ is idempotent
\begin{equation}\label{pzricI}
 \underline{\widehat{Ric}}_I(\underline{z})=Q(\underline{z})\underline{\widehat{Ric}}_I(\underline{z})
\end{equation}
holds for all $\underline{z}\in\cM(K^n)$.
We view $\underline{\widehat{Ric}}_I(\underline{z})$ as a trace free 
part  of $\widehat{Ric}(\underline{z})$. In fact 
$$
\langle \underline{z},\underline{\widehat{Ric}}_I(\underline{z})\rangle=0
$$
holds. We consider relation \eqref{g:12} to be the analogous relation in the smooth case.
The following scaling relation is valid
\begin{equation*}
 \underline{\widehat{Ric}}_I(\lambda\underline{z})=\lambda^{(n-4)/2}\underline{\widehat{Ric}}_I(\underline{z}),
\end{equation*}
which is the same as for $\underline{Ein}(\underline{z})$ itself.
The main result of this subsection is the 
\begin{theorem}\label{theo:EI}
 Let $\underline{z}_0\in\cM(K^n)$. The following conditions are equivalent.
\begin{enumerate}
 \item{$\underline{z}_0$ is an Einstein metric of type I.}
\item{$\underline{z}_0$ satisfies 
\begin{equation}\label{PI}
 \widehat{\underline{Ein}}_I(\underline{z}_0)=0.
\end{equation}
}
\item{$\underline{z_0}$ is a critical point of the scale invariant function
\begin{equation*}
 F_I(\underline{z})=\frac{1}{||\underline{z}||^{(n-2)/2}}\cR(\underline{z})
 =\cR\left(\frac{1}{||\underline{z}||}\underline{z}\right).
 \end{equation*}
}
\item{$\underline{z}_0$ is a critical point of the function
\begin{equation}\label{actionI}
A_I(\underline{z})=\cR(\underline{z})-\frac{\kappa_I}{2}\,||\underline{z}||^2. 
\end{equation}
}
\item{$\underline{z}_0$ is a critical point of the total scalar curvature $\cR(\underline{z})$ 
restricted to
$\cM_{r=||\underline{z}_0||}(K^n)$.}
\item{$\underline{z}_0$ is a solution of the Euler-Lagrange equation, where the Lagrange function is the 
total scalar curvature and the constraint is the function
\begin{equation}\label{fC}
C(\underline{z})=\frac{1}{2}||\underline{z}||^2-\frac{1}{2}||\underline{z}_0||^2.
\end{equation}
}
\item{ $\underline{z}_0$ satisfies
\begin{equation*}
\underline{Ein}(\underline{z}_0)
=\frac{2}{n}
\frac{\langle \underline{v}(\underline{z}_0),\underline{Ein}(\underline{z}_0)\rangle}
{V(\underline{z}_0)}\;\underline{z_0}.
\end{equation*}
}
\item{$\underline{z}_0$ satisfies
\begin{equation*}
 \underline{Ein}(\underline{z}_0) =\frac{n-2}{n}
\frac{||\underline{Ein}(\underline{z}_0)||^2}
{\cR(\underline{z}_0)}\;\underline{z_0}
 \end{equation*}
in case $\cR(\underline{z}_0)\neq 0$.
}
\item{The relation 
\begin{equation}\label{ricci301}
 \cR(\underline{z}_0)^2=\frac{4}{(n-2)^2}||\underline{z}_0||^2||\underline{Ein}(\underline{z}_0)||^2
\end{equation}
is valid.
}
\end{enumerate}
\end{theorem}
This theorem compares with a well known result for Einstein spaces attributed to Hilbert, 
see e.g.\ Theorem 4.21 in \cite{Besse}. The smooth analogue of $F_I$ is given in \eqref{g:18}.
\begin{proof}
Condition \eqref{PI} is equivalent to the condition that $\underline{Ein}(\underline{z}_0)$ is a multiple of 
the vector $\underline{z}_0$ by the definition of $Q(\underline{z}_0)$. Thus (1) and (2) are equivalent.
(3) is equivalent to the condition that the gradient of $F_I(K^n,\underline{z})$ vanishes at $\underline{z}_0$. 
But 
\begin{equation*}
\underline{\nabla}F_I(\underline{z})=\frac{1}{||\underline{z}||^{(n-2)/2}}
\widehat{\underline{Ein}}_I(\underline{z})
\end{equation*}
so (1) and (3) are equivalent by Proposition \ref{pro:EI1}. The equivalence of (1) and (4) is also clear.
The condition (5) states that 
$$
\langle \underline{u}, \underline{Ein}(\underline{z}_0)\rangle=0
$$
holds for every $\underline{u}\in T_{\underline{z}_0}\cM_{||\underline{z}_0||}(K^n)$. Now every such 
$\underline{u}$ can be written in the form $\underline{u}=Q(\underline{z}_0)\underline{x}$ for some 
vector $\underline{x}$, since $Q(\underline{z}_0)$ is the orthogonal projection onto the tangent space
$T_{\underline{z}_0}\cM_{||\underline{z}_0||}(K^n)$. Therefore (5) is equivalent to 
$$
\langle Q(\underline{z}_0)\underline{x}, \underline{Ein}(\underline{z}_0)\rangle=0 
\quad\mbox{for all }\quad \underline{x}\in \R^{\bar{n}}.
$$
So (2) and (5) are equivalent. As for condition (6) let $\kappa_I$ denote the Lagrange 
multiplier associated to the constraint $\eqref{fC}$. The Euler-Lagrange equation is then 
just \eqref{ricci1}. Thus (1) and (6) are equivalent. Alternatively the equivalence of (5) and (6) is a 
consequence of the Euler- Lagrange variational principle, by which $\kappa_I$ is also fixed. (7) and (8) are 
consequences of 
(1) due to Proposition \ref{prop:11}. The converse is obvious.
It remains to prove the equivalence of (1) and (9). By Schwarz inequality and \eqref{scaling4}
$$
 \cR(\underline{z})^2\le 
\frac{4}{(n-2)^2}||\underline{z}||^2||\underline{Ein}(\underline{z})||^2
$$
holds with equality if and only if the vectors $\underline{Ein}(\underline{z})$ and 
$\underline{z}$ are collinear. If $\underline{Ein}(\underline{z}_0)=0$, that is if 
$(K^n,\underline{z}_0)$ is 
Einstein-flat, then $\underline{z}_0$ is an Einstein metric of type I and \eqref{ricci301} holds. 
But if $\underline{Ein}(\underline{z}_0)\neq 0$ then \eqref{ricci1} holds with $\kappa_I\neq 0$.
Thus (1) and (9) are equivalent.
\end{proof}
The r\^{o}les of $C(\underline{z})$ and $\cR(\underline{z})$ can almost be interchanged. 
In fact with 
$$
\widetilde{\widetilde{\cM}}_{\rho}(K^n)=\{\underline{z}\in\cM(K^n)\,|\,\cR(\underline{z})=\rho\}
$$ 
we have
\begin{corollary}\label{cor:EI}
Assume $\underline{z}_0\in \cM(K^n)$ is such that 
$\underline{Ein}(\underline{z}_0)\neq 0$.
Then the following two conditions are equivalent.
\begin{enumerate}
 \item{$\underline{z}_0$ is an Einstein metric of type I.}
\item{$\underline{z}_0$ is a critical point of $C(\underline{z})$ on 
$\widetilde{\widetilde{\cM}}_{\rho=R(\underline{z}_0)}(K^n)$.}
\end{enumerate}
\end{corollary}
\begin{proof}
Take $\cR(\underline{z})-\cR(\underline{z}_0)$ as a constraint and take $1/\kappa_I$ to be the 
the Lagrange multiplier. With $C(\underline{z})$ as the Lagrange function the claim follows.
\end{proof}

\vspace{0.5cm}
\subsection{Einstein metrics of type II}~~\\
\begin{lemma}\label{lem:EII}
If $(K^n,\underline{z}_0)$ is a p.l. Einstein space of type II, which is not Einstein-flat, then 
\begin{equation}\label{EII1neq0}
 \langle\underline{v}(\underline{z}_0),\underline{Ein}(\underline{z}_0)\rangle\neq 0.
\end{equation}
\end{lemma}
\begin{proof}
Take the scalar product of \eqref{vol1} with $\underline{Ein}(\underline{z}_0)$, which gives
\begin{equation}\label{EII1neq}
 ||\underline{Ein}(\underline{z}_0)||^2-\kappa_{II}
\langle\underline{v}(\underline{z}_0),\underline{Ein}(\underline{z}_0)\rangle=0.
\end{equation}
Assume now that \eqref{EII1neq0} is not valid, that is its left hand side vanishes. But then 
$\underline{Ein}(\underline{z}_0)=0$, contradicting the assumption.
\end{proof}
The following result is analogous to the one given in Proposition \ref{pro:EI1}.
\begin{proposition}\label{pro:EII1}
Let $(K^n,\underline{z}_0)$ be a p.l. Einstein space of type II, which is not Einstein-flat. Then $\kappa_{II}$ is 
given 
in terms of the total scalar curvature $\cR(\underline{z})$ \eqref{totalscalar1} and the volume as 
$\kappa_{II}^{(1)}(\underline{z}_0)$ where
\begin{equation}\label{ricci31}
\kappa_{II}^{(1)}(\underline{z})
=\frac{n-2}{n}\overline{\cR}(\underline{z})
\end{equation}
with $\overline{\cR}(\underline{z})$ denoting the {\rm average scalar curvature}, see \eqref{avtotalscalar}. $\kappa_{II}$ is also given as $\kappa_{II}^{(2)}(\underline{z}_0)$ 
where
\begin{equation*}
\kappa_{II}^{(2)}(\underline{z})=
\frac{\langle \underline{v}(\underline{z}),
\underline{Ein}(\underline{z})\rangle}{||\underline{v}(\underline{z})||^2}.
\end{equation*}
$\kappa_{II}^{(2)}(\underline{z})$ is well defined for all $\underline{z}\in\cM(K^n)$.
Finally $\kappa_{II}$ is also given as $\kappa_{II}^{(3)}(\underline{z}_0)$ where
\begin{equation}\label{ricci312}
\kappa_{II}^{(3)}(\underline{z})=\frac{||\underline{Ein}(\underline{z})||^2}
{\langle\underline{v}(\underline{z}),\underline{Ein}(\underline{z})\rangle}
\end{equation}
which is well defined for all $\underline{z}\in \cM(K^n)$ with 
$\langle\underline{v}(\underline{z}),\underline{Ein}(\underline{z})\rangle\neq 0$.
\end{proposition}
Observe that $\kappa_{II}^{(1)}(\underline{z}_0)=\kappa_{II}^{(2)}(\underline{z}_0)=\kappa_{II}^{(3)}(\underline{z}_0)=0$,
if $(K^n,\underline{z}_0)$ is Einstein-flat, that is if $\underline{Ein}(\underline{z}_0)=0$ holds. 
Conversely, if $\underline{Ein}(\underline{z}_0)\neq 0$, then also 
$\langle\underline{v}(\underline{z}_0),\underline{Ein}(\underline{z}_0)\rangle\neq 0$ by \eqref{EII1neq}.
Observe also that like the volume $V(\underline{z})$ its gradient $\underline{v}(\underline{z})$ never vanishes on $\cM(K^n)$ 
due to \eqref{volgrad3}.
\begin{proof}
Using the fact that $V(\underline{z})$ is homogeneous of degree $n/2$, the proof of the first claim follows the 
same line as the proof of Proposition \ref{pro:EI1}.
As for the second claim take the scalar product of \eqref{vol1} with $\underline{v}(\underline{z}_0)$.
As for the third claim, since $\underline{Ein}(\underline{z}_0)\neq 0$ by assumption, we have
$\langle\underline{v}(\underline{z}_0),\underline{Ein}(\underline{z}_0)\rangle\neq 0$ by Lemma \ref{lem:EII}. 
So the third claim follows from \eqref{EII1neq}.
\end{proof}
Set 
\begin{equation*}
\widetilde{\cM}_{v}(K^n)=\{\underline{z}\in\cM(K^n)\;|\; V(\underline{z})=v\}
\end{equation*}
with $v>0$. Also set
\begin{equation}\label{RicII}
\widehat{\underline{Ein}}_{II}(\underline{z})=\underline{Ein}(\underline{z})
-\frac{n-2}{n}\overline{\cR}(\underline{z})\,\underline{v}(\underline{z})
\end{equation}
which again is trace free, that is 
$$
\langle\underline{z},\widehat{\underline{Ein}}_{II}(\underline{z})\rangle=0
$$
or equivalently 
\begin{equation*}
Q(\underline{z})\widehat{\underline{Ein}}_{II}(\underline{z})=\widehat{\underline{Ein}}_{II}(\underline{z})
\end{equation*}
is valid for all $\underline{z}\in\cM(K^n)$.
The scaling behavior is
\begin{equation*}
 \underline{\widehat{Ric}}_{II}(\lambda\underline{z})=
 \lambda^{(n-4)/2}\underline{\widehat{Ric}}_{II}(\underline{z}),
\end{equation*}
which is the same as for $\underline{Ein}(\underline{z})$ itself.
\begin{theorem}\label{theo:EII}
 Let $\underline{z}_0\in\cM(K^n)$. The following conditions are equivalent.
\begin{enumerate}
 \item{$\underline{z}_0$ is an Einstein metric of type II.}
\item{$\widehat{\underline{Ein}}_{II}(\underline{z}_0)=0$.}
\item{$\underline{z_0}$ is a critical point of the scale invariant function
\begin{equation*}
 F_{II}(\underline{z})=\frac{1}{V(\underline{z})^{(n-2)/n}}\cR(\underline{z})
=\cR\left(\frac{1}{V(\underline{z})^{2/n}}\underline{z}\right).
 \end{equation*}
 }
\item{$\underline{z}_0$ is a critical point of the function
\begin{equation}\label{actionII}
A_{II}(\underline{z})=\cR(\underline{z})-\kappa_{II}V(\underline{z}). 
\end{equation}
}
\item{$\underline{z}_0$ is a critical point of the total scalar curvature $R(\underline{z})$ 
restricted to $\widetilde{\cM}_{v=V(\underline{z}_0)}(K^n)$.}
\item{$\underline{z}_0$ is a solution of the Euler-Lagrange equation, where the Lagrange function is the 
total scalar curvature and the constraint is the volume function $V(\underline{z})$.
}
\item{$\underline{z}_0$ satisfies 
\begin{equation*}
 \underline{Ein}(\underline{z}_0)=\frac{\langle \underline{v}(\underline{z}_0),
\underline{Ein}(\underline{z}_0)\rangle}{||\underline{v}(\underline{z}_0)||^2}\underline{v}(\underline{z}_0).
\end{equation*}
}
\item{$\underline{z}_0$ satisfies
\begin{equation*}
 ||\underline{Ein}(\underline{z}_0)||^2||\underline{v}(\underline{z}_0)||^2=
\langle \underline{v}(\underline{z}_0),
\underline{Ein}(\underline{z}_0)\rangle^2.
\end{equation*}
}
\end{enumerate}
\end{theorem}
Again this theorem compares with Theorem 4.21 in \cite{Besse} and $F_{II}$ compares with \eqref{g:18}.
\begin{proof}
In view of Proposition \ref{pro:EII1} (1) and (2) are equivalent as are (1) and (7). 
The condition (3) says that the gradient of $F_{II}(\underline{z},K^n)$ should vanish for 
$\underline{z}=\underline{z}_0$. But 
\begin{equation}\label{FIIgrad}
\underline{\nabla} F_{II}(\underline{z})=\frac{1}{V(\underline{z})^{(n-2)/n}}
\widehat{\underline{Ein}}_{II}(\underline{z})
\end{equation}
which shows the equivalence of (2) and (3). 
The equivalence of (1) and (4) 
is also clear as is the equivalence of both (5) and (6) with (4). The equivalence of (1) with (7) is also clear.
The equivalence of (7) with (8) follows from Schwarz inequality.
\end{proof}
The analogue of these two actions \eqref{actionI} and \eqref{actionII} is in the smooth case given by \eqref{g:24}.
In dimensions $n=3,4$ the p.l.\ version of the Einstein equations without a cosmological term, that is the equation
$\underline{Ein}(\underline{z})=0$, has already been given and discussed by Regge \cite{Regge}.
The analogue to the relations \eqref{ricci3}, \eqref{ricci30}, \eqref{ricci300}, \eqref{ricci31}, 
\eqref{ricci312} and 
\eqref{ricci312} in the smooth case is given by relation \eqref{g:9}.

\vspace{0.5cm}

\subsection{Examples}\label{subsec:ex}~~\\
First we provide an example of a p.l. Einstein-flat space. It is modeled on the $n$-torus $T^n$, which we recall 
is obtained as follows.
On $\R^n$ the group $\Z^n$ acts in a natural way as a transformation group. The $n$-torus is then just the quotient space 
$\R^n/\Z^n$. Consider a triangulation of $\R^n$ which is invariant under $\Z^n$. Such a triangulation is easy 
to construct. Indeed it 
suffices to construct a suitable triangulation on an $n$-cube. This is done by induction on $n$. For $n=1$, 
the closed unit interval $[0,1]$, 
declare the two endpoints to be vertices and in addition consider the barycenter, that is the point $1/2$, to be the additional vertex.
The intervals $[0,1/2]$ and $[1/2,1]$ are the two 1-simplexes. Now consider an $n$-cube. 
For each of its $2n$ faces, which are $(n-1)$ cubes,
by the induction assumption we can construct a triangulation. Add the barycenter $\rv_b$ of the $n$-cube 
as a new vertex. In addition to the 
simplexes on the faces by fiat the new simplexes are of the form $\sigma^k\cup\{\rv_b\}$, 
where $\sigma^k$ is any simplex in any of the faces of 
the $n$-cube. This completes the induction step.
This triangulation of $\R^n$ induces a triangulation of $T^n$, denoted by $\mathcal{T}^n$. The edge lengths are of course induced by the 
euclidean metric on $\R^n$.
\begin{example}\label{ex:0}
 $\mathcal{T}^n$ for $n\ge 3$ is a p.l. Einstein space of both types, which in addition is Ricci-flat. 
\end{example}
\begin{proof}
It is clear that the deficit angle at any $\sigma^{n-2}$ vanishes, thus not only the average scalar curvature vanishes but 
also $\underline{Ein}$ due to \eqref{totalscalar3}.  
\end{proof}
Also for a given pseudomanifold $K^n$ and given $\sigma^{l}\in K^n$, let $N_k(\sigma^{l};K^n),\;k>l$ 
denote the number of $k$-simplexes
which contain $\sigma^{l}$, that is $N_k(\sigma^{l};K^n)=\sharp\{\sigma^{k}\in K^n|\sigma^{k}\supset\sigma^{l}\}$.
For any $a>0$, let $\underline{a}$ denote the metric by which 
$z_{\sigma^1}=a$ holds for all $\sigma^1$, that is all edge lengths are equal to $\sqrt{a}$, and 
$||\underline{a}||=\sqrt{n_1(K^n)}a$.

Finally let $N_k(K^n)$ denote the total number of $k-$simplexes in $K^n$. The following lemma provides a 
sufficient condition for a pseudomanifold $K^n$ to carry an Einstein metric.
\begin{lemma}\label{lem:einstein}
 Let $K^n$ be such that all numbers $N_n(\sigma^{n-2};K^n)$ are equal $(=\overline{N}_1)$ as well 
as all
 $N_{n-2}(\sigma^1;K^n)\;(=\overline{N}_2)$. Then $\underline{a}$ is an Einstein metric on $K^n$ of type I. 
If in addition all 
$N_{n}(\sigma^1,K^n)$ are equal $(=\bar{N}_3)$, then $\underline{a}$ is also an Einstein metric of type II.
\end{lemma}
\begin{proof}
 Observe first that $\partial^{\sigma^1}|\sigma^{n-2}|(\underline{a})$ is independent of 
 $\sigma^1$ and $\sigma^{n-2}$ with 
$\sigma^1\in\sigma^{n-2}$ (and of course zero otherwise). It depends only on $n$ and $a$, 
is of the form  
$g(n-2)a^{(n-4)/2}$, where $g(n)$ will be given below, see \eqref{partialvol}.
Similarly the dihedral angle $(\sigma^{n-2},\sigma^{n})(\underline{a})$ only depends on $n$, 
$(\sigma^{n-2},\sigma^{n})(\underline{a})=\phi(n)$ and 
is also given below, see \eqref{dihedralangle}.
Therefore 
$$
Ein_{\sigma^1}(K^n,\underline{a})=\overline{N}_2\left(1-\overline{N}_1 
\phi(n)\right)\,g(n-2)\,a^{\frac{n-4}{2}}
$$
holds and is independent of $\sigma^1$. The first part of the lemma follows. 
As for the second part, the last assumption means that 
$$
v_{\sigma^1}(K^n,\underline{a})=\overline{N}_3\,g(n)\,a^{\frac{n-2}{2}},
$$
which is independent of $\sigma^1$. The second part of the lemma follows.
\end{proof}
As an application we obtain
\begin{example}\label{ex:1}
$(\partial\sigma^{n+1},\underline{a}); n\ge 3$ is an Einstein space of both types. 
With the choice 
\begin{equation*}
\kappa_I=\frac{1}{a}\,Ein_{\sigma^1}(\partial\sigma^{n+1},\underline{a})
\end{equation*}
the condition in \eqref{ricci1} is satisfied. The volume of any euclidean $n$-simplex with equal 
edge lengths $\sqrt{a}$ is known, 
see \cite{RHBuchholz}, 
\begin{equation}\label{Voleinstein}
 V(\sigma^{n},\underline{a})=\frac{a^{n/2}}{n!}\sqrt{\frac{n+1}{2^n}}.
\end{equation}
Since $\partial\sigma^{n+1}$ contains $n+2\; n$-simplexes this gives
\begin{equation*}
 V(\partial\sigma^{n+1},\underline{a})=\frac{(n+2)a^{n/2}}{n!}\sqrt{\frac{n+1}{2^n}}.
\end{equation*}
The dihedral angle is given as \cite{ParksWills}
\begin{equation}\label{dihedralangle}
\phi(n)= (\sigma^{n-2},\sigma^n)=\frac{1}{2\pi}
\arccos\frac{1}{n}.
\end{equation}
Also $g(n)$, defined in the proof of Lemma \ref{lem:einstein}, is given as 
\begin{equation}\label{partialvol}
 g(n)=\frac{1}{(n+1)!}\sqrt{\frac{n+1}{2^{n}}}.
\end{equation}
The total scalar curvature equals
\begin{equation}\label{einsteincurv}
 \cR(\partial\sigma^{n+1},\underline{a})=a^{(n-2)/2}
 \left(\begin{matrix}n+2\\2
\end{matrix}\right)\frac{1}{(n-2)!}\sqrt{\frac{n-1}{2^{n-2}}}
\left(1-\frac{3}{2\pi}\arccos\frac{1}{n}\right)
\end{equation}
and is in particular positive.
The Einstein vector field at $\underline{a}$ is given by
\begin{equation}\label{einsteinric}
Ein_{\sigma^1}(\partial\sigma^{n+1},\underline{a})=\frac{(n-2)}{2a}\cR(\partial\sigma^{n+1},\underline{a})
\qquad \mbox{for all}\quad \sigma^1.
\end{equation}
Similarly 
\begin{equation}\label{einsteinrvol}
 v_{\sigma^1}(\partial\sigma^{n+1},\underline{a})
 =\frac{1}{(n+1)a}V(\partial\sigma^{n+1},\underline{a}) \qquad \mbox{for all}\quad \sigma^1
\end{equation}
holds, so with the choice 
\begin{equation*}
\kappa_{II}=\frac{n-2}{n}\frac{\cR(\partial\sigma^{n+1},\underline{a})}{V(\partial\sigma^{n+1},\underline{a})}
\end{equation*}
compare \eqref{ricci31}, the condition in \eqref{vol1} is satisfied.
To sum up, $(\partial\sigma^{n+1},\underline{a})$ is a p.l. Einstein space of both types. 
\end{example}
The proofs of \eqref{einsteincurv} -- \eqref{einsteinrvol} will be given in Appendix \ref{app:1}.
Additional examples for p.l. Einstein spaces seem hard to come by. Thus we do not know 
whether Lemma \ref{lem:einstein} allows for other examples.
Also we do not know, whether, there exist pseudomanifolds 
having Einstein metrics of one type only.
Recall for comparison that the spheres $S^n$ with the round metric have constant sectional 
curvature and hence are Einstein manifolds, see e.g. \cite{Besse} p. 44.  
At present we do not know of any pseudomanifold $K^n$, which does not carry an Einstein metric (of either type).
However, there is the much weaker result, by which there are p.l. spaces, which are not p.l. Einstein spaces of 
type I.
\begin{example}\label{notE}
Consider any subdivision of 
$(\partial \sigma^{n+1},\underline{a})$ with the following property: It has at least one $1-$simplex,
whose star is contained in the interior of a euclidean  $n-$simplex of $\partial \sigma^{n+1}$. Such subdivisions 
can easily be constructed. Any such subdivision is Ricci-flat at at least one $1-$simplex  but not Ricci-flat. 
\end{example}


\section{Einstein flows}\label{sec:ricciflows}

In this section we will define Einstein flows and normalized Einstein flows.
In what follows, $K^n$ with $n\ge 3$ will be fixed, and again we will mostly leave $K^n$ out of the notation.

Given a pseudomanifold $K$, we would like to find an Einstein metric 
$\underline{z}_0$ of type I or II on $K^n$ through a flow on $\cM(K^n)$.

By proposition \ref{theo:EI} (3) a first idea would be 
to look for a minimum of $\cR(\underline{z})^2$. However, due to the scaling 
behavior \eqref{scaling3}
$$
\lim_{\lambda\downarrow 0}\cR(\lambda\underline{z})=0
$$
holds for any $\underline{z}\in \cM(K^n)$. In order to avoid this situation, one has to 
make a restriction. One possibility is to look for variations, which e.g. preserve the volume.

This will bring us to the concept of {\it normalized Einstein flow}, for an introduction see e.g. 
\cite{ChowLuNi}.
We recall that in the smooth case, a compact Einstein metric is a fixed point of the 
normalized Ricci flow, by which the volume is  preserved. 
Conversely, any fixed point of the normalized Ricci flow is an 
Einstein metric.

We start by defining the (unnormalized) Einstein flow equation as the gradient flow
\begin{equation}\label{Flow1}
\dot{\underline{z}}(t)=-2 \underline{Ein}(\underline{z}(t))
\end{equation}
on $\cM(K^n)$. Here and in what follows, $\;\dot{}\;$ denotes taking the time derivative $\rd/\rd t$. 
The factor 2 is in order to conform with the 
standard convention in the smooth case and 
can obviously be changed by a suitable rescaling of the time.

Flows starting at an Einstein metric of type I are particularly simple.
\begin{proposition}\label{prop:21}
Let $(K^n,\underline{z}_0)$ be a p.l. Einstein space of type I. Then
\begin{equation}\label{Flow121}
 \underline{z}(t)=f_n(t)\underline{z}_0
\end{equation}
is a solution to the flow equation \eqref{Flow1} with initial condition $\underline{z}(t=0)=\underline{z}_0$.
In particular $\underline{z}(t)$ is an Einstein metric of the same type.

With the notation
$$
\kappa(t)= \kappa_{I}(K^n,\underline{z}(t))
$$
and $\kappa=\kappa(t=0)$ for the initial value, the following relation is valid
\begin{equation}\label{kflow}
\kappa(t)=f_n(t)^{(n-6)/2}\kappa.
\end{equation}
For $n\neq 6$ $f_n(t)$ is of the form
\begin{equation}\label{fflow}
 f_n(t)=(1+(n-6)\kappa t)^{-\frac{2}{n-6}}
\end{equation}
valid for all $0\le t<\infty$ if $(n-6)\kappa>0$ and for all $0\le t <-((n-6)\kappa)^{-1}$ if 
$(n-6)\kappa<0$,
while
\begin{equation}\label{f6}
 f_6(t)=\e^{-2\kappa t}\quad \mbox{for all} \quad 0\le t<\infty.
\end{equation}
\end{proposition}
Thus when $(n-6)\kappa<0$, then $f_n$ tends to zero in finite time if $n<6$ and to infinity in finite time when
$n>6$.
So far we have not been able to prove an analogous result for Einstein metrics of type II.
\begin{proof}
Make \eqref{Flow121} an Ansatz. Then \eqref{ricci1} in case of type I  combined with \eqref{Ricscale} 
give the differential equation
\begin{equation}
\dot{f}_{n}(t)=-2\kappa f_{n}(t)^{(n-4)/2},
\end{equation}
which may be transformed into 
\begin{align*}
 \rd(f_n)^{-\frac{(n-6)}{2}}&=(n-6)\kappa\, \rd t,\: n\neq 6,\\\nonumber
 \rd\ln f_6&=-2\kappa\, \rd t.
\end{align*}
Combined with the initial condition $f_{n}(t=0)=1$ this easily gives \eqref{fflow} and \eqref{f6}.
\eqref{kflow} follows from \eqref{ricci3} and \eqref{ricci31} and the scaling laws \eqref{volumescaling} and 
\eqref{scaling3}.
\end{proof}
In general, for a solution of \eqref{Flow1}
\begin{equation*}
 \frac{\rd}{\rd t}||\underline{z}(t)||^2=
-2(n-2)\cR(\underline{z}(t))
\end{equation*}
follows by \eqref{scaling4}. Since we only consider $n\ge 3$, under this flow 
$||\underline{z}(t)||$ increases if
$\cR(\underline{z}(t))<0$, decreases if $\cR(\underline{z}(t))>0$ and is stationary  
at times $t$ for which $\cR(\underline{z}(t))=0$. 

\vspace{0.5cm}
\subsection{Normalized Einstein flows of type I}~~\\
The {\it \underline{first} normalized Einstein flow of type I} is defined by the differential equation
\begin{equation}\label{Flow3}
\dot{\underline{z}}(t)=-2 \widehat{\underline{Ein}}_I(\underline{z}(t)).
\end{equation}
The {\it \underline{second} normalized Einstein flow of type I} is defined by the differential equation
\begin{equation}\label{Flow32}
\dot{\underline{z}}(t)=-2\underline{Ein}(\underline{z}(t))
+\frac{4}{n}\frac{\langle\underline{v}(\underline{z}(t)),\underline{Ein}(\underline{z}(t))\rangle}
{V(\underline{z}(t))}\underline{z}(t).
\end{equation}
The {\it \underline{third} normalized  Einstein flow of type I} is defined by the differential 
equation
\begin{equation}\label{Flow33}
\dot{\underline{z}}(t)=-2\underline{Ein}(\underline{z}(t))
+\frac{2}{n-2}\frac{||\underline{Ein}(\underline{z}(t))||^2}
{\cR(\underline{z}(t))}\underline{z}(t).
\end{equation}
By Theorem \ref{theo:EI} a p.l.\ Einstein metric of type I is a fixed point of all these
flow equations, whence the name flows of type I.
By standard results for non-linear differential equations all these equations have solutions 
$\underline{z}(t)$ for all small $t$ as long as the initial condition $\underline{z}(0)$ lies in $\cM(K^n)$.
For the third flow \eqref{Flow33} one has to assume $\cR(\underline{z}(0))\neq 0$ in addition.
\begin{proposition}\label{prop:3}
\begin{itemize}
\item{For any solution $\underline{z}(t)$ of the flow equation \eqref{Flow3} $||\underline{z}(t)||$ 
and $V(\underline{z}(t))$ 
are constant.
}
\item{For any solution $\underline{z}(t)$ of the flow equation \eqref{Flow32} the volume
$V(\underline{z}(t))$ is constant. 
}
\item{For any solution $\underline{z}(t)$ of the flow equation 
\eqref{Flow33} the total scalar curvature
$\cR(\underline{z}(t))$ is constant.
}
\end{itemize}
\end{proposition}
Before we turn to a proof, we use this result to elaborate on the differential equation \eqref{Flow3}.
$\underline{\widehat{Ric}}(\underline{z}(t))$ can only become singular, when $\underline{Ein}(\underline{z}(t))$ 
becomes singular. By \eqref{totalscalar3} in turn this is only possible if 
$\partial^{\sigma^1}|\sigma^{n-2}|(\underline{z}(t))$ becomes 
singular for at least one pair $\sigma^1\subset \sigma^{n-2}$. Therefore by \eqref{dervol} 
the r.h.s.\ of \eqref{Flow3} can only become singular when at least one of the volumes 
$|\sigma^{n-2}|(\underline{z}(t))$ 
tends to zero. The two other flow equations may be discussed similarly.
\begin{proof}
\eqref{pxx}, \eqref{pzricI} and \eqref{Flow3} give
$$
\frac{\rd}{\rd t}\,\langle\underline{z}(t),\underline{z}(t)\rangle=
2\langle\underline{z}(t),\dot{\underline{z}}(t)\rangle=
-4\langle\underline{z}(t),\underline{\widehat{Ric}}_I(\underline{z}(t))\rangle=0,
$$
as well as 
\begin{align*}
\frac{\rd}{\rd t}\,V(\underline{z}(t))&=\langle\dot{\underline{z}}(t),\underline{v}(\underline{z}(t))\rangle\\\nonumber
&=-2\langle\underline{Ein}(\underline{z}(t)), \underline{v}(\underline{z}(t))\rangle
+\frac{4}{n}\frac{\langle\underline{Ein}(\underline{z}(t)), \underline{v}(\underline{z}(t))\rangle}{V(\underline{z}(t))}
\langle\underline{v}(\underline{z}(t)),\underline{z}(t)\rangle=0,
\end{align*}
which proves the first claim. As for the second claim
\begin{align*}
\frac{\rd}{\rd t}\,V(\underline{z}(t))&=\langle \dot{\underline{z}}(t),\underline{v}(\underline{z}(t))\rangle\\
&=-2\langle \underline{Ein}(\underline{z}(t), \underline{v}(\underline{z}(t))\rangle
+\frac{4}{n}\frac{\langle\underline{v}(\underline{z}(t),\underline{Ein}(\underline{z}(t))\rangle}
{V(\underline{z}(t))}\langle\underline{z}(t),\underline{v}(\underline{z}(t))\rangle =0.
\end{align*}
We have used \eqref{volgrad3}. The last claim also follows by arguments, which by now are standard
$$
\frac{\rd }{\rd t}\cR(\underline{z}(t))
=-2||\underline{Ein}(\underline{z}(t))||^2
+\frac{4}{n-2}\frac{||\underline{Ein}(\underline{z}(t))||^2}{\cR(\underline{z}(t))}
\langle \underline{z}, \underline{Ein}(\underline{z}(t))\rangle=0.
$$
\end{proof}
This result states that with initial condition $\underline{z}(0)$
\begin{itemize}
\item{
the first normalized Einstein flow of type I is a flow in 
$\cM_{r=||\underline{z}(0)||}(K^n)$,} 
\item{
the second normalized Einstein flow of type II is a flow  
in $\widetilde{\cM}_{v=V(\underline{z}(0))}(K^n)$,
}
\item{the third normalized Einstein flow of type III 
is a flow in 
$\widetilde{\widetilde{\cM}}_{\rho=\cR(\underline{z}(0))}(K^n)$.
}
\end{itemize}

If the initial condition $\underline{z}(0)$ happens to be such that $(K^n,\underline{z}(0))$ is Einstein-flat at 
a 1-simplex $\sigma^1$, then 
\begin{itemize}
 \item{$z_\sigma^1$ and hence also $l_{\sigma^1}$ increase for all small $t$ if the total scalar curvature 
 $\cR(\underline{z}(0))$ is strictly positive.}
 \item{$z_\sigma^1$ and hence also $l_{\sigma^1}$ decrease for all small $t$ if the total scalar curvature 
 $\cR(\underline{z}(0))$ is strictly negative.}
 \item{$z_\sigma^1(t)$ and hence also $l_{\sigma^1}(t)$ are stationary at $t=0$, if $\cR(\underline{z}(0))=0$.}
 \end{itemize}
The following example in 3 dimensions illustrates this point.
For $n=3$ by \eqref{totalscalar3} the Einstein vector field takes the form
\begin{equation}\label{ric=3}
 Ein_{\sigma^1}(K^{n=3},\underline{z})
 =\left(1-\sum_{\sigma^3\supset \sigma^1}\left(\sigma^1,\sigma^3\right)\right)\frac{1}{2\sqrt{z_{\sigma^1}}}.
\end{equation}
\begin{example}\label{ex:lengthincr}
Let $(K^{n=3\;\prime},\underline{z}^\prime)$ be a subdivision of $(K^{n=3},\underline{z})$. Since the
deficit angle around any 
$\sigma^{1\;\prime}\in \Theta^1(K^{n=3\;\prime},\underline{z}^\prime)$ 
vanishes - see the discussion of relation \eqref{totalscalar10} -
 $(K^{n=3\;\prime},\underline{z}^\prime)$ is Einstein-flat at such $\sigma^{1\;\prime}$.
\end{example}
Of special interest is the case $(K^{n=3},\underline{z})=(\partial \sigma^4,\underline{a})$, a p.l. Einstein 
space with positive total scalar curvature. We now make a specific choice of the 
subdivision, namely we take $(K^{n=3\,\prime},\underline{a}^\prime)$ to be the 
{\rm barycentric subdivision}. This has the advantage that the symmetry of 
$(\partial \sigma^4,\underline{a})$ 
under the group of permutations of the vertices is preserved.
\begin{proposition}\label{prop:barysub}
Under the barycentric subdivision  $(K^{n=3\,\prime},\underline{a}^\prime)$ of the 
p.l.\ Einstein space $(\partial \sigma^4,\underline{a})$, we have 
\begin{equation}\label{barysub}
\widehat{Ric}_{\sigma^{1\;\prime}}(K^{n=3\,\prime},\underline{a}^\prime)=
\begin{cases} < 0\quad\mbox{for} \quad \sigma^{1\;\prime}\in \Theta^1(K^{n=3\;\prime},\underline{a}^\prime)\\
  > 0\quad\mbox{for} \quad \sigma^{1\;\prime}\notin \Theta^1(K^{n=3\;\prime},\underline{a}^\prime).
\end{cases}
\end{equation}
Accordingly the lengths 
$\underline{z}_{\sigma^1}^\prime(t)$ increase or decrease 
for all small $t$ under the flow \eqref{Flow3} with initial 
condition $\underline{z}^\prime(t=0)=\underline{a}^\prime$.
Moreover 
$\widehat{Ric}_{\sigma^{1\;\prime}}(K^{n=3\,\prime},\underline{a}^\prime)$ takes the 
same value for all 
$\sigma^{1\;\prime}\notin \Theta^1(K^{n=3\;\prime},\underline{a}^\prime)$.
 In particular the barycentric subdivision of  the 
p.l. Einstein space $(\partial \sigma^4,\underline{a})$ is not a p.l. Einstein space.

\end{proposition}
Observe that for a barycentric subdivision 
$a^\prime_{\sigma^{1\,\prime}}
=\frac{1}{4}a_{\sigma^1}$ 
when $\sigma^{1\;\prime}\preceq\sigma^1$.

\begin{proof}
 The last claim follows by the symmetry of the barycentric subdivision mentioned above.
 Also this common value has to be positive by the first case in \eqref{barysub} and since $|\underline{z}(t)|^2$ is conserved 
 under the flow \eqref{Flow3} with initial condition $\underline{z}^\prime$ or 
 equivalently by the tracelessness of 
 $\underline{\widehat{Ric}}$, that is
 $\langle\underline{z}^\prime, \underline{\widehat{Ric}}(\underline{z}^\prime)\rangle=0$.
\end{proof}

By the scaling properties of the quantities involved, we immediately obtain
the following 
\begin{proposition}\label{prop:31}
Let $\underline{z}(t)$ be a solution of any of the three flow equations \eqref{Flow3}, \eqref{Flow32} and 
\eqref{Flow33}
with initial condition $\underline{z}(0)$ and 
let $\lambda>0$ be arbitrary.  
Then $\underline{z}^\lambda(t)=\lambda\underline{z}(\lambda^{(n-6)/2}t)$ is also a solution of the 
same flow equation with initial condition $\lambda\underline{z}(0)$.
\end{proposition}
Returning to \eqref{Flow1} and \eqref{Flow3}, by a proper scaling in space and time one 
can obtain a solution of the normalized Einstein flow 
from one of the Einstein flow itself. Indeed, let $\underline{z}(t)$ be a 
solution of the Einstein flow and set $\underline{\tilde{l}}(\tilde{t})=c(t)\underline{z}(t)$ with 
\begin{align*}
c(t)=\e^{\frac{2}{n}\int_0^t \cR(\underline{z}(s))\rd s},\qquad \tilde{t}(t)=\int_0^t c(s)\rd s.
\end{align*}
Then $\underline{\tilde{z}}(\tilde{t})$ is a solution of the first normalized Einstein flow. 
The proof is just as in the smooth case, see e.g. \cite{ChowLuNi}.

\begin{theorem}\label{theo:4}~~\\
\begin{itemize}
\item{Under the first normalized Einstein flow \eqref{Flow3} the total scalar curvature is a strictly decreasing 
function 
of $t$ except when $\underline{z}(t)$ is an Einstein metric of type I
\begin{equation}\label{RdotI}
\dot{\cR}(\underline{z}(t))= -2\,\Big|\Big|\underline{\widehat{Ric}}_I(\underline{z}(t))\Big|\Big|^2.
\end{equation}
}
\item{Let $\underline{z}(t)$ be a solution of the third normalized Einstein flow.
Assume $\underline{z}(t)$ is not an Einstein metric of type I and $\cR(\underline{z}(t))\neq 0$. Then 
$||\underline{z}(t)||$ is strictly increasing at $t$ if $\cR(\underline{z}(t))>0$ and  
strictly decreasing at $t$ if $\cR(\underline{z}(t))<0$.
}
\end{itemize}
\end{theorem}
Below, see Lemma \ref{lem:4}, we will see that $R(\underline{z})$ remains bounded, when $||\underline{z}||$ 
stays bounded.  
\begin{proof}
Taking derivative of $R(\underline{z}(t))$ w.r.t. $t$ and using \eqref{Flow3} gives
\begin{align}\label{derr}
\dot{\cR}(\underline{z}(t))&=\langle\dot{\underline{z}}(t),\underline{Ein}(\underline{z}(t))\rangle\\\nonumber&= 
-2\langle\underline{\widehat{Ric}}_I(\underline{z}(t)),\underline{Ein}(\underline{z}(t))\rangle
\end{align}
and \eqref{RdotI} follows by \eqref{ricchat}. As for the second claim we calculate
\begin{align}\label{zflow}
\frac{\rd}{\rd t}||\underline{z}(t)||^2&=-4\langle\underline{Ein}(\underline{z}(t)),\underline{z}(t)\rangle
+\frac{8}{n-2}\frac{||\underline{Ein}(\underline{z}(t)||^2}{\cR(\underline{z}(t))}||\underline{z}(t)||^2
\\\nonumber
&=\frac{8}{(n-2)\cR(\underline{z}(t))}\left(||\underline{Ein}(\underline{z}(t)||^2||\underline{z}(t)||^2-
\langle\underline{Ein}(\underline{z}(t)),\underline{z}(t)\rangle^2\right)
\end{align}
and so the claim follows by Schwarz inequality.
\end{proof}
Observe that for given $t$ the right hand side of \eqref{RdotI} vanishes if and only if 
$\underline{z}(t)$ is an Einstein metric of type I, see Theorem \ref{theo:EI}. 
The same holds for the r.h.s.\ of \eqref{zflow}.
An immediate consequence is the 
\begin{corollary}\label{corr:4}
 Let $\underline{z}_0$ be an Einstein metric of type I. For the first normalized Einstein flow of type I to 
approach $\underline{z}_0$ from the initial 
condition $\underline{z}(t=0)\neq \underline{z}_0$ it is necessary that 
\begin{itemize}
 \item{$||\underline{z}(t=0)||=||\underline{z}_0||$}
 \item{$\cR(\underline{z}(t=0))>\cR(\underline{z}_0)$}
\end{itemize}
holds.
\end{corollary}
Because any Einstein metric of type I is a fixed point of any of these three flows, an approach 
to such a metric can only be asymptotic due to the following lemma.
\begin{lemma}\label{lem:asymp}
An approach to an Einstein metric of type I under any of these flows can at most be asymptotic.
\end{lemma}
\begin{proof}
It suffices to consider the first flow, for the other two flows the proof is similar with some adaptions. 
Assume that under 
the flow $\underline{z}(t)$, where $\underline{z}(0)$ is not an Einstein metric, an Einstein metric 
$\underline{z}_0$ 
is reached in finite 
time, say $\underline{z}(T)=\underline{z}_{0}$. 
Consider the time reversed flow defined by $\underline{z}_{\mbox{\it rev}}(t)=\underline{z}(T-t),\; 0\le t\le T$.
It satisfies the time reversed flow equation
\begin{equation}\label{trev}
\dot{\underline{z}}_{\mbox{\it rev}}(t)=2\underline{\widehat{Ric}}_I(\underline{{z}}_{\mbox{\it rev}}(t))
\end{equation}
and starts at $\underline{z}(T)$. But this leads to a contradiction, since 
$\dot{\underline{z}}_{\mbox{\it rev}}(t)$ vanishes for $t=0$ by \eqref{trev} and therefore for all $0\le t\le T$ by the 
uniqueness of solutions of \eqref{trev} for given initial condition.
\end{proof}
We consider the first normalized flow to be the most promising one for further studies. 
Indeed, in combination with 
condition (5) of Theorem \ref{theo:EI} we have 
\begin{corollary}\label{corr:41}
 Assume $\underline{z}_{\min}$ is a local minimum of $\cR(\underline{z})$ on the set 
 $\cM_{r=||\underline{z}_{\min}||}(K^n)$. Then $\underline{z}_{\min}$ is an Einstein metric of type I.
Assume in addition that $\underline{z}_{\min}$ is non-degenerate.
 Then there is a neighborhood $\cU(\underline{z}_{\min})$ in 
$\cM_{r=||\underline{z}_{\min}||}(K^n)$ of $\underline{z}_{\min}$, 
such that the flow \eqref{Flow3} starting there (but away from $\underline{z}_{\min}$) will stay there and 
approach $\underline{z}_{\min}$ asymptotically.
\end{corollary}
\begin{proof}
The first part follows from the following observation. $\cR(\underline{z}(t))$ is strictly decreasing as long as 
the traceless Einstein vector field is non-vanishing. Since $\underline{z}_{\min}$ is a local minimum, the 
traceless Einstein vector field must be vanishing there and this is equivalent for $\underline{z}_{\min}$ to be an 
Einstein metric of the type I. If $\underline{z}_{\min}$ is non-degenerate, there is a 
neighborhood of $\underline{z}_{\min}$, which does not contain another Einstein metric, that is any other critical 
point of $\cR(\underline{z})$ on $\cM_{r=||\underline{z}_{\min}||}(K^n)$. Now we again use the 
fact that $\cR(\underline{z}(t))$ is strictly decreasing away from an Einstein metric. The last claim follows by the 
previous lemma.
\end{proof}
A further immediate consequence of \eqref{derr} and Proposition \ref{prop:3} is the relation
\begin{equation}\label{derr1}
\cR(\underline{z}(t))= \cR(\underline{z}(0))-2\,\int_0^t
\Big|\Big|\underline{\widehat{Ric}}_I(\underline{z}(s))\Big|\Big|^2\,\rd s
\end{equation}
for a solution of the normalized Einstein flow equation up to time $t$. 

Each of the three quantities
\begin{align}\label{deltaI1}
 \Delta_I^{(1)}(\underline{z})&=
\Big|\Big|\underline{\widehat{Ric}}_I(\underline{z})\Big|\Big|^2,\\\nonumber 
\Delta_I^{(2)}(\underline{z})&=
\Big|\Big|\underline{Ein}(\underline{z})
-\frac{2}{n}\frac{\langle\underline{v}(\underline{z}),\underline{Ein}(\underline{z})\rangle}
{V(\underline{z})}\underline{z}\Big|\Big|^2,\\\nonumber
\Delta_I^{(3)}(\underline{z})&=\Big|\Big|\underline{Ein}(\underline{z})
-\frac{2}{n-2}\frac{||\underline{Ein}(\underline{z})||^2}
{\cR(\underline{z})}\underline{z}\Big|\Big|^2
\end{align}
can be viewed as a measure for how much $\underline{z}$ deviates from an Einstein metric of type I
on $K^n$.

The next result states that the total scalar curvature decreases at least linearly in 
time as long as one stays strictly away from an Einstein metric.
\begin{corollary}\label{corr:1}
For given initial condition $\underline{z}(0)$, which is not an Einstein metric, let the solution 
of the first flow equation exist up to time $T>0$.  Then there is a constant $c>0$, depending on $\underline{z}(0)$ and 
$T$ only, such that 
\begin{equation*}
\cR(\underline{z}(t))\le \cR(\underline{z}(0))-2ct
\end{equation*}
holds for all $0\le t\le T$.
\end{corollary}
\begin{proof}
By assumption, by the continuity of $s\mapsto \underline{z}(s)$ and the continuity of the 
maps $\underline{z}\mapsto R(\underline{z})$ and 
$\underline{z}\mapsto\underline{Ein}(\underline{z})$ 
$$
c=\inf_{0\le s\le T} \;\Delta_I^{(1)}(\underline{z}(s))
$$
is strictly positive. The claim now follows from \eqref{derr1}.
\end{proof}

The $\Delta_I^{(i)}(\underline{z}),\; (i=1,2,3)$ satisfy the scaling relation 
\begin{equation}
\Delta_I^{(i)}(\lambda \underline{z})=\lambda^{n-4}\Delta_{I}^{(i)}(\underline{z}).
\end{equation}
Therefore their infimum
\begin{equation*}\label{derr3}
 N_I^{(i)}(r)=\inf_{\underline{z}\in \cM_r(K^n)}\Delta_I^{(i)}(\underline{z})
\end{equation*}
satisfy the scaling relation $N_I^{(i)}(\lambda r)=\lambda^{n-4}N_I^{(i)}(r)$.
We have the obvious result 
\begin{lemma}\label{lem:21}
For $K^n$ to have an Einstein metric of type I, it is necessary that $N_I^{(i)}(r)=0$ holds for all $i$ 
and some $r>0$ (and hence all $r$).
\end{lemma}
From \eqref{derr1} we derive the {\it a priori} estimate
\begin{equation}\label{derr4}
 \cR(\underline{z}(t))\le \cR(\underline{z}(0))-2tN_I^{(1)}(||\underline{z}(0)||)
\end{equation}
for any initial condition $\underline{z}(0)$.

In Appendix \ref{app:2} we prove the next lemma. 
It provides smoothness properties of the total scalar curvature and 
the Einstein vector field, some of which we already have used.
\begin{lemma}\label{lem:3}
The volume, the total scalar curvature, the Einstein vector field and the traceless Einstein vector fields are 
smooth functions of the metric $\underline{z}\in\cM(K^n)$.
\end{lemma}
Therefore by standard results from the theory of differential equations, 
for given initial condition $\underline{z}(0)\in \cM(K^n)$ there is a unique solution 
$\underline{z}(t)\in \cM(K^n)$ 
to the normalized Einstein equation of type I for 
$0\le t\le T\, (T>0)$.
We will choose $T$ to be maximal, thus allowing for $T=\infty$ and then 
$T$ depends on the initial condition condition only,
$T=T(\underline{z}(0))$. Observe that the solution can not run to infinity, since 
$||\underline{z}(t)||=||\underline{z}(0)||$ for all $t$. 

So if we assume $T<\infty$, then 
$\underline{z}(T)\in \partial\cM(K^n)\cap\overline{\cM_{||\underline{z}_0||}(K^n)}$.

If we could prove that the vector field $\underline{\widehat{Ric}}_I(\underline{z})$ 
is ``tangential'' to the boundary $\partial \cM(K^n)$ for $\underline{z}\in\partial \cM(K^n)$, and
hence actually ``tangent'' to $\partial\cM_{||\underline{z}||}(K^n)$, 
then the flow could never leave $\cM(K^n)$ and we would have arrived at a contradiction that $T$ 
is finite. So we turn to a more detailed analysis, first of the total scalar curvature, the Ricci vector 
field and the traceless Einstein vector fields near the the boundary and then to an analysis of the 
boundary itself.
The following bounds are obvious
\begin{equation*}
 |\sigma^k|\le c_k||\underline{z}||^{k/2},\qquad 0<(\sigma^{n-2},\sigma^n)\le 1.
\end{equation*}
The $c_k<\infty$ are universal constants.
Let $N_k(K^n)$ denote the number of $k$-simplexes in $K^n$, and
$$
N_{k,l}(K^n)=\max_{\sigma^k\in K^n}\;\sharp(\sigma^l\,:\,\sigma^l\supset\sigma^k),
$$
the maximum number of times a $k$-simplex is the face of an $l$-simplex.
\begin{lemma}\label{lem:4}
The bounds 
\begin{align*}
V(\underline{z})&\le c_n N_n(K^n)||\underline{z}||^{n/2}\\\nonumber
 |\cR(K^n,\underline{z})|&\le c_{n-2}N_{n-2}(K)N_{n-2,n}(K^n)||\underline{z}||^{(n-2)/2}
\end{align*}
are valid. 
\end{lemma}
As an immediate consequence we obtain the following result: With 
$$
\cR_{\min}(r)=\min_{\underline{z}\;:\;||\underline{z}||=r}\cR(\underline{z})
$$
the estimate 
\begin{equation*}
 \cR_{\min}(r)\ge -c_{n-2}N_{n-2}(K^n)N_{n-2,n}(K^n)r^{(n-2)/2}
\end{equation*}
is valid. Combining this with the estimate \eqref{derr4} we obtain the 
\begin{proposition}
If $K^n$ is such that $N_I(K^n,r)>0$, then a flow starting at $\underline{z}_0$ cannot be continued beyond
the time $T$ with 
\begin{align*}
T&\le \frac{1}{2N_I(K^n,||\underline{z}_0||)}
\left(\cR(K^n,\underline{z}_0)-\cR_{\min}(K^n,||\underline{z}_0||) \right)\\\nonumber
&\le \frac{1}{2N_I(K^n,||\underline{z}_0||)}
\left(\cR(K^n,\underline{z}_0)-c_{n-2}N_{n-2}(K^n)N_{n-2,n}(K^n)||\underline{z}_0||^{(n-2)/2}\right).
\end{align*}
\end{proposition}

\vspace{0.5cm}

\subsection{Normalized Einstein flows of type II}~~\\
In this subsection we provide an alternative definition of a normalized Einstein flow and which is closely 
related to the 
concept of a p.l. Einstein space of type II. For this definition we invoke the gradient 
$\underline{v}$ of the volume.  

By definition the {\it \underline{first} normalized Einstein flow equation of type II} is given as 
\begin{equation}\label{newflow}
\dot{\underline{z}}(t)=-2 \underline{\widehat{Ric}}_{II}(\underline{z}(t)),
\end{equation}
see \eqref{RicII}.
The right hand side of \eqref{newflow} equals
\begin{equation*}
-2V(\underline{z})^{(n-2)/n}\underline{\nabla} F_{II}(\underline{z}),
\end{equation*}
see \eqref{FIIgrad}.

By definition the {\it \underline{second} normalized Einstein flow equation of type II} is given as 
\begin{equation}\label{newflow1}
\dot{\underline{z}}(t)=-2 \underline{Ein}(\underline{z}(t))
+2\frac{\langle \underline{v}(\underline{z}(t)),
\underline{Ein}(\underline{z}(t))\rangle}{||\underline{v}(\underline{z}(t))||^2}\underline{v}(\underline{z}(t)).
\end{equation}
By definition the {\it \underline{third} normalized Einstein flow equation of type II} is given as 
\begin{equation}\label{newflow11}
\dot{\underline{z}}(t)=-2 \underline{Ein}(\underline{z}(t))+
2\frac{||\underline{Ein}(\underline{z}(t))||^2}
{\langle\underline{v}(\underline{z}(t)),\underline{Ric}(\underline{z}(t))\rangle}.
\end{equation}
By Proposition \ref{pro:EII1} an Einstein metric of type II is a fixed point under all these flows.
Set 
$$
\cM^0(K^n)=\{\underline{z}\in\cM(K^n)\,|\, 
\langle\underline{v}(\underline{z}),\underline{Ric}(\underline{z})\rangle= 0\}.
$$
In analogy to Proposition \ref{prop:3}  we have 
\begin{proposition}\label{prop:newflow}
Under the flow \eqref{newflow} $||\underline{z}(t)||$ is constant while under the flow \eqref{newflow1} 
$V(\underline{z}(t))$ is constant. Under the flow \eqref{newflow11} $\cR(\underline{z}(t))$ is constant as long as 
$\underline{z}(t)\notin\cM^0(K^n)$.
\end{proposition}
Recall that unless the Einstein metric $\underline{z}_0$ of type II is Einstein-flat, one has 
$\langle\underline{v}(\underline{z}_0),\underline{Ein}(\underline{z}_0)\rangle\ne 0$, that is 
$\underline{z}_0\notin\cM^0(K^n)$. Therefore by continuity there is a whole neighborhood of $\underline{z}_0$, 
which does not meet $\cM^0(K^n)$.
\begin{proof}
The first claim follows from the tracelessness of $\underline{\widehat{Ric}}_{II}$, since
$$
\frac{\rd}{\rd t}\,\langle\underline{z}(t),\underline{z}(t)\rangle=
2\langle\underline{z}(t),\dot{\underline{z}}(t)\rangle
=-4 \langle\underline{z}(t),\underline{\widehat{Ric}}_{II}(\underline{z}(t))\rangle
=0.
$$
As for the second claim 
\begin{align*}
\frac{\rd}{\rd t}\,V(\underline{z}(t))&=\langle \dot{\underline{z}}(t),\underline{v}(\underline{z}(t))\rangle\\
&=-2\langle \underline{Ein}(\underline{z}(t), \underline{v}(\underline{z}(t))\rangle
+2\frac{\langle\underline{v}(\underline{z}(t),\underline{Ein}(\underline{z}(t))\rangle}
{||\underline{v}(\underline{z}(t))||^2}
\langle\underline{v}(\underline{z}(t)),\underline{v}(\underline{z}(t))\rangle =0.
\end{align*}
The proof of the last claim is analogous and will be left out.
\end{proof}
In analogy to Proposition \ref{prop:31} we have
\begin{proposition}\label{prop:newflow0}
Let $\underline{z}(t)$ be a solution of one of the three flow equations \eqref{newflow}, \eqref{newflow1} and
\eqref{newflow11}
with initial condition $\underline{z}(0)$ and 
let $\lambda>0$ be arbitrary.  
Then $\underline{z}^\lambda(t)=\lambda\underline{z}(\lambda^{(n-6)/2}t)$ is a solution of the same flow equation
with initial condition $\lambda\underline{z}(0)$.
\end{proposition}
Set 
\begin{equation*}
 \widehat{\cR}(\underline{z})=\frac{\cR(\underline{z})}{V(\underline{z})^{\frac{n-2}{n}}}
=\cR\left(\frac{1}{V(\underline{z})^{\frac{2}{n}}}\underline{z}\right)\;,
\end{equation*}
a scale invariant quantity. In analogy to Theorem \ref{theo:4} we have 
\begin{theorem}\label{theo:6}
Under the flow \eqref{newflow} $\widehat{\cR}(\underline{z}(t))$ is decreasing and strictly decreasing except 
at an Einstein metric of type II 
\begin{equation}\label{rren}
 \frac{\rd}{\rd t}\widehat{\cR}(\underline{z}(t))=
-\frac{2}{V(\underline{z}(t))^{\frac{n-2}{n}}}
\Big|\Big|\underline{\widehat{Ric}}_{II}(\underline{z}(t))\Big|\Big|^2.
\end{equation}
Under the flow \eqref{newflow1} $\cR(\underline{z}(t))$ is strictly decreasing except 
at an Einstein metric of type II due to 
\begin{equation}\label{dotrII}
 \frac{\rd}{\rd t}\cR(\underline{z}(t))=-2\frac{1}{||\underline{v}(\underline{z}(t))||^2}
\left(||\underline{v}(\underline{z}(t))||^2||\underline{Ein}(\underline{z}(t))||^2-
\langle\underline{v}(\underline{z}(t)),\underline{Ein}_{II}(\underline{z}(t))\rangle^2\right)
\end{equation}
and Schwarz inequality.
\end{theorem}
The comment after Corollary \ref{corr:4} carries over to the present situation: Since any Einstein metric of 
type II is a 
fixed point of the flow \eqref{newflow}, any approach 
to such a metric under this flow can at most be asymptotic. 
\begin{proof}
A short calculation gives
$$
\frac{\rd}{\rd t}\widehat{\cR}(\underline{z}(t))=\frac{1}{V(\underline{z}(t))^{\frac{n-2}{n}}}
\langle\underline{\dot{z}}(t),
\underline{\widehat{Ric}}_{II}(\underline{z}(t))\rangle
$$
and \eqref{rren} follows by inserting the flow equation \eqref{newflow}. \eqref{dotrII} follows by an easy 
calculation, so the 
last claim is a consequence of Schwarz inequality and statement (8) in Theorem \ref{theo:EII}.
\end{proof}
In analogy to Corollary \ref{corr:4} we have the 
\begin{corollary}\label{corr:5}
 Let $\underline{z}_0$ be an Einstein metric of type II. For the first normalized Einstein flow of type II to 
approach $z_0$ from the initial 
condition $\underline{z}(t=0)\neq \underline{z}_0$ it is necessary that 
\begin{itemize}
 \item{$||\underline{z}(t=0)||=||\underline{z}_0||$}
 \item{$\widehat{\cR}(\underline{z}(t=0))>\widehat{\cR}(\underline{z}_0)$}
\end{itemize}
holds.

For the second normalized Einstein flow of type II to 
approach $z_0$ from the initial 
condition $\underline{z}(t=0)\neq \underline{z}_0$ it is necessary that 
\begin{itemize}
 \item{$V(\underline{z}(t=0))=V(\underline{z}_0)$}
 \item{$\cR(\underline{z}(t=0))>\cR(\underline{z}_0)$}
\end{itemize}
holds.
\end{corollary}
In analogy to Corollary \ref{corr:41} we have 
\begin{corollary}\label{corr:51}~~\\
\begin{itemize}
\item{Let $\underline{z}_{\min}$ be a local minimum of $\widehat{\cR}(\underline{z})$ on the set 
$\cM_{r=||\underline{z}_{\min}||}(K^n)$.
Then there is a neighborhood $\cU(\underline{z}_{\min})$ in $\cM_{r=||\underline{z}_{\min}||}(K^n)$ of 
$\underline{z}_{\min}$, 
such that the flow \eqref{newflow} starting there will stay there and approach $\underline{z}_{\min}$.
}
\item{Let $\underline{z}_{\min}$ be a local minimum of $\cR(\underline{z})$ on the set 
$\widetilde{\cM}_{v=V(\underline{z}_{\min})}(K^n)$.
Then there is a neighborhood $\cU(\underline{z}_{\min})$ in $\widetilde{\cM}_{v=V(\underline{z}_{\min})}(K^n)$ of 
$\underline{z}_{\min}$, such that the flow \eqref{newflow1} starting there will stay there 
and approach $\underline{z}_{\min}$.}
\end{itemize}
\end{corollary}
In analogy to \eqref{deltaI1}, each of the quantities
\begin{align*}
\widehat{\Delta}_{II}^{(1)}(\underline{z})&
=\Big|\Big|\underline{\widehat{Ric}}_{II}(\underline{z})\Big|\Big|^2\\\nonumber
\widehat{\Delta}_{II}^{(2)}(\underline{z})&=
\Big|\Big|\underline{Ein}(\underline{z})
-\frac{\langle \underline{v}(\underline{z}),
\underline{Ein}(\underline{z})\rangle}{||\underline{v}(\underline{z})||^2}
\underline{v}(\underline{z})\Big|\Big|^2
\\\nonumber
\widehat{\Delta}_{II}^{(3)}(\underline{z})&=
\Big|\Big| \underline{Ein}(\underline{z})-
\frac{||\underline{Ein}(\underline{z})||^2}
{\langle\underline{v}(\underline{z}),\underline{Ein}(\underline{z})\rangle} \underline{v}(\underline{z})\Big|\Big|^2
\end{align*}
is a measure for how much the metric $\underline{z}$ deviates from an Einstein 
metric of type II on $K^n$.  


\section{Second variation of the total scalar curvature at the boundary 
\newline
of the equilateral $4-$simplex.}\label{sec:var}

In this section we will analyze the behavior of $\cR(\partial\sigma^4,\underline{z})$, where
$\underline{z}$ is close to the Einstein metric $\underline{a}$, by computing the second variation. 
Similar calculations have been carried out on the double tetrahedron in \cite{CGY}.

As a preparation we discuss the general case, namely the second order variation of the 
total scalar curvature at an arbitrary p.l. Einstein 
space $(K^n,\underline{z}_E)$ (of the first or second type). 
Then we consider the variation at fixed fourth moment of the edge lengths, that is $||\underline{z}||^2$ 
stays fixed. Finally we determine the variation at fixed volume $V(\underline{z})$.
For a corresponding discussion in the smooth case see \cite{Schoen}.

The pseudomanifold $\partial\sigma^4$ has five vertices and
ten $1-$simplexes. The relations $||\underline{a}||^2=10a^2$ and 
$a\sum_{\sigma^1}u_{\sigma^1}=\langle \underline{a},\underline{u}\rangle$ will often be used without 
explicit mentioning. Any nonempty set of vertices defines a simplex in $\partial\sigma^4$. 
Therefore any $1-$ simplex is contained in three $3-$ simplexes.
The automorphism group {\it Aut}$(\partial \sigma^4)$ is easily seen to 
be isomorphic to $\bS_5$, the permutation 
group of 5 elements. In fact, any restriction $s\in\;${\it Aut}$(\partial \sigma^4)$ to the five 
vertices is just a permutation. Conversely any permutation $s$ of the vertices can uniquely be extended 
to an automorphism of the pseudomanifold $\partial\sigma^4$. Any automorphism automatically 
extends to a metric preserving automorphism of $(\partial\sigma^4,\underline{a})$. We shall refer to 
this observation as the {\it symmetry} (of $(\partial\sigma^4,\underline{a})$).
There is a representation $s\mapsto T(s)$ of 
${\it Aut}(\partial \sigma^4)$ into $GL(10,\R)$ given as $(T(s)\underline{z})_{\sigma^1}=
z_{s^{-1}\sigma^1}$, where we assume the set of $1-$simplexes to be ordered in some way. $T(s)$ is just a 
permutation matrix and $\det T(s)^2=1$ holds. Observe that the set of $10\times 10$ permutation matrices
defines a representation of the permutation group $\bS_{10}$, a much greater set.

Furthermore consider the following linear real representation $s\rightarrow O(s)$ of $Aut(\partial\sigma^4)$ 
 on $\R^{10}$ given as $(O(s)\underline{x})_{\sigma^1}=x_{s^{-1}\sigma^1}$.
 Since obviously $||O(s)\underline{x}||=||\underline{x}||$, this representation is also orthogonal. 
 It leaves $\cM(\partial\sigma^4)$ 
 and each $\cM_{||\underline{a}||}(\partial\sigma^4)$ invariant. In other words $Aut(\partial\sigma^4)$ 
 acts as a transformation group on each of these spaces. 
 $(\partial\sigma^4,\underline{a})$ is the only fixed point on $\cM_{||\underline{a}||}(\partial\sigma^4)$.
 
Let $\underline{z}(t)$ be a local differentiable one-parameter family of edge lengths squared and 
let $\;\dot{}\;$ 
denote differentiation 
w.r.t. $t$. By \eqref{Reggevar} 
\begin{equation}\label{var5}
 \ddot{\cR}=\sum_{\sigma^{n-2}}\dot{\delta}(\sigma^{n-2})\dot{|\sigma^{n-2}|}+
 \sum_{\sigma^{n-2}}\delta(\sigma^{n-2})\ddot{|\sigma^{n-2}|}.
\end{equation}
The obvious relations
\begin{align*}
 \dot{|\sigma^{n-2}|}&=\sum_{\rho^1}\dot{z}_{\rho^1}\partial^{\rho^1}|\sigma^{n-2}|,\\\nonumber
 \ddot{|\sigma^{n-2}|}&=\sum_{\rho^1}\ddot{z}_{\rho^1}\partial^{\rho^1}|\sigma^{n-2}|
 +\sum_{\sigma^1, \rho^1}\dot{z}_{\sigma^1}\dot{z}_{\rho^1}\partial^{\sigma^1}\partial^{\rho^1}|\sigma^{n-2}|
\end{align*}
give the general relation
\begin{align}\label{var52}
\ddot{\cR}&=-\sum_{\rho^1}\sum_{\sigma^{n-2}}\sum_{\sigma^n\supset\sigma^{n-2}}\dot{(\sigma^{n-2},\sigma^n)}
\dot{z}_{\rho^1}\partial^{\rho^1}|\sigma^{n-2}|\\\nonumber
&\quad+\sum_{\rho^1,\sigma^{n-2}}\delta(\sigma^{n-2})
\left(\ddot{z}_{\rho^1}\partial^{\rho^1}|\sigma^{n-2}|
+\sum_{\sigma^1}\dot{z}_{\sigma^1}\dot{z}_{\rho^1}\partial^{\sigma^1}\partial^{\rho^1}|\sigma^{n-2}|
\right).
\end{align}
In the concrete case of 
$(\partial\sigma^4,\underline{a})$ we are able to determine the explicit form of the second order 
variation.

\subsection{Second variation of the total scalar curvature with fixed
fourth 
\newline
moment of the edge lengths.} 

\begin{theorem}\label{theo:locmax}
 The second order variation of the total scalar curvature on $\cM_{||\underline{a}||}(\partial\sigma^4)$ 
at $(\partial\sigma^4,\underline{a})$ is negative definite. 
Therefore $(\partial\sigma^4,\underline{a})$ is a local maximum on 
$\cM_{||\underline{a}||}(\partial\sigma^4)$.
\end{theorem}
For a comparison with the smooth case, see \cite{Schoen}, p. 125.

The remainder of this subsection is devoted to a proof of this theorem.
So we specialize \eqref{var52} to $(\partial\sigma^4,\underline{a})$, 
such that in particular $n=3$, and we will take recourse to \eqref{var5} rather than \eqref{var52}.
Also we make the choice 
\begin{equation}\label{var1}
 \underline{z}(t)=||\underline{a}||\frac{\underline{a}+t\underline{u}}{||\underline{a}+t\underline{u}||},
\end{equation}
a vector with $||\underline{z}(t)||=||\underline{a}||$ and $\underline{z}(t=0)=\underline{a}$. 
$\underline{u}$ is arbitrary and $-\varepsilon< t< \varepsilon$ with  $\varepsilon>0$ 
sufficiently small. 

Set $F_{\underline{u}}(t)=\cR(\underline{z}(t))$, 
so the object of interest is $\ddot{F}_{\underline{u}}(t=0)$. Observe that $F_{\underline{u}=0}(t)$ 
is a constant, namely $R(\underline{a})$. The following lemma will be useful.
\begin{lemma}\label{lem:F}
For any $\lambda$ the relation 
\begin{equation}\label{var1301}
F_{\underline{u}+\lambda\underline{a}}(t)
=F_{\underline{u}}(t^\prime)
\end{equation}
with $t^\prime=t/(1+\lambda t)$ is valid, such that 
\begin{equation}\label{var1302}
\frac{\rd^2}{\rd t^2}F_{\underline{u}+\lambda\underline{a}}(t=0)=
\frac{\rd^2}{\rd t^{\prime^2}}F_{\underline{u}}(t^\prime=0)
\end{equation}
holds. In particular 
$F_{\underline{u}}(t)$ is constant if $P(\underline{a})\underline{u}=\underline{u}$ and 
\begin{equation}\label{var1303}
\ddot{F}_{\underline{u}}(t=0)=\ddot{F}_{(\1-P(\underline{a}))\underline{u}}(t=0).  
\end{equation}
holds for general $\underline{u}$.
\end{lemma}
\begin{proof}
\eqref{var1301} follows from the trivial relation 
$$
\frac{\underline{a}+t(\underline{u}+\lambda\underline{a})}
{||\underline{a}+t(\underline{u}+\lambda\underline{a})||}
=\frac{\underline{a}+t^\prime \underline{u}}{||\underline{a}+t^\prime \underline{u}||}.
$$
\eqref{var1302} follows from a short calculation using \eqref{var1301} and the relation 
$$
\dot{F}_{\underline{u}}(t=0)=\dot{\cR}(\underline{z}(t=0))=\langle\dot{\underline{z}}(t=0),
\underline{\nabla} \cR(\underline{a})\rangle
 =\langle\dot{\underline{z}}(t=0),k\underline{a}\rangle=0,
$$
which holds due to \eqref{var3} and since $(\partial\sigma^4,\underline{a})$  is a p.l. Einstein space.
The last claims follows from \eqref{var1301} by making the choice 
$\lambda=-\langle \underline{a},\underline{u}\rangle/||\underline{a}||^2$, such that 
$\underline{u}+\lambda\underline{a}=0$ and by using \eqref{var1302}.
\end{proof}
For the computation of \eqref{var5} the derivatives therein have to be calculated.
The relation
\begin{equation}\label{var2}
 \dot{\underline{z}}(t)=||\underline{a}||\left(\frac{1}{||\underline{a}+t\underline{u}||}\underline{u}
 -\frac{\langle(\underline{a}+t\underline{u}),\underline{u}\rangle}{||\underline{a}+t\underline{u}||^3}
 \left(\underline{a}+t\underline{u}\right)\right)
\end{equation}
gives
\begin{equation}\label{var3}
\dot{\underline{z}}(t=0)=(\1-P(\underline{a}))\underline{u}
\end{equation}
and therefore the first variation of the total scalar curvature at $t=0$ vanishes as should be, since
\begin{equation}\label{var4}
 \dot{\cR}(\underline{z}(t=0))=\langle\dot{\underline{z}}(t=0),\underline{\nabla}\cR(\underline{a})\rangle
 =\langle\dot{\underline{z}}(t=0),k\underline{a}\rangle=0.
\end{equation}
Relation \eqref{var3} gives
\begin{equation}\label{var6}
\dot{|\sigma^1|}(\underline{z}(t=0))=\dot{\sqrt{z_{\sigma^1}}}(t=0)
=\frac{\dot{z}_{\sigma^1}(t=0)}{2\sqrt{a}}
=\frac{1}{2\sqrt{a}}\left((\1-P(\underline{a}))\underline{u}\right)_{\sigma^1}.
\end{equation}
Taking the derivative of \eqref{var2} gives
\begin{equation}\label{var7}
 \ddot{\underline{z}}(t)=||\underline{a}||\left(-2\frac{\langle\underline{a},
 \underline{u}\rangle}{||\underline{a}+t\underline{u}||^3}\underline{u}
 -\frac{\langle\underline{u},\underline{u}\rangle}{||\underline{a}+t\underline{u}||^3}
 \left(\underline{a}+t\underline{u}\right)
 +\frac{3\langle(\underline{a}+t\underline{u}),\underline{u}\rangle^2}{||\underline{a}+t\underline{u}||^5}
  (\underline{a}+t\underline{u})\right)
\end{equation}
and hence
\begin{equation}\label{var8}
 \ddot{\underline{z}}(t=0)=-2
 \frac{\langle\underline{a},\underline{u}\rangle}{||\underline{a}||^2}\underline{u}
 -\frac{\langle\underline{u},\underline{u}\rangle}{||\underline{a}||^2}\underline{a}
 +\frac{3\langle\underline{a},\underline{u}\rangle^2}{||\underline{a}||^4}\underline{a}.
\end{equation}
The relation 
\begin{equation}\label{var9}
 \ddot{\sqrt{z_{\sigma^1}}}(t=0)=\frac{1}{2}\frac{\ddot{z}_{\sigma^1}(t=0)}{z_{\sigma^1}(t=0)^{1/2}}
 -\frac{1}{4}\frac{\dot{z}_{\sigma^1}^2(t=0)}{z_{\sigma^1}(t=0)^{3/2}}
\end{equation}
implies
\begin{equation}\label{var10}
\ddot{|\sigma^1|}(\underline{z}(t=0))
=-\frac{\langle\underline{a},\underline{u}\rangle u_{\sigma^1}}{\sqrt{a}||\underline{a}||^2}
-\frac{1}{2}\frac{\sqrt{a}\langle\underline{u},\underline{u}\rangle}{||\underline{a}||^2}
+\frac{3}{2}\frac{\sqrt{a}\langle \underline{u},P(\underline{a})\underline{u}\rangle}
{||\underline{a}||^2}
-\frac{1}{4}\frac{\left(\left(\1-P(\underline{a})\right)\underline{u}\right)_{\sigma^1}^2}{a^{3/2}},
\end{equation}
where 
$
\langle\underline{a},\underline{u}\rangle^2/||\underline{a}||^2
=\langle \underline{u},P(\underline{a})\underline{u}\rangle
$
has been used. 
Now we are able to provide the second term on the r.h.s.\ of \eqref{var5} in the present context.
A short calculation gives the following quadratic form 
\begin{equation}\label{var102}
\langle\underline{u},Q_2\underline{u}\rangle= \sum_{\sigma^{1}}\delta(\sigma^{1})\ddot{|\sigma^{1}|}(\underline{a})=
 \left(1-\frac{3}{2\pi}\arccos\frac{1}{3}\right)\sum_{\sigma^{1}}\ddot{|\sigma^{1}|}(\underline{a})
\end{equation}
that is 
\begin{equation}\label{var103}
 Q_2=-\frac{3}{4a^{3/2}}\left(1-\frac{3}{2\pi}\arccos\frac{1}{3}\right)\left(\1-P(\underline{a})\right).
\end{equation}
Use has been made of the symmetry by which all $\delta(\sigma^{1})$ are equal. Note that this result is 
in agrement with relation \eqref{var1302}. Actually by this relation one may make the replacement 
$\underline{u}\rightarrow (\1-P(\underline{a}))\underline{u}$ in \eqref{var10} providing an easier proof of 
\eqref{var102}. Below, see \eqref{V8}, 
a similar argument will be used to simplify an otherwise lengthier calculation.

The term $\dot{\delta}(\sigma^1)$ in \eqref{var5} (with $n=3$) is harder to come by. By the chain rule
\begin{equation}\label{var110}
 \dot{\delta}(\sigma^1)=\sum_{\rho^1}\partial^{\rho^1}\delta(\sigma^1)\dot{z}_{\rho^1},
\end{equation}
that is
\begin{equation}\label{var11}
\dot{\delta}(\sigma^1)=-\sum_{\rho^1}M^{\sigma^1\;\rho^1}\dot{z}_{\rho^1}
\end{equation}
with the $10\times 10$ matrix
\begin{equation}\label{var12}
 M^{\sigma^1\;\rho^1}=-\partial^{\rho^1}\delta(\sigma^1)=
 \sum_{\sigma^3\supset\sigma^1,\;\sigma^3\supset\rho^1}\partial^{\rho^1}(\sigma^1,\sigma^3).
\end{equation}
We claim that $\partial^{\rho^1}(\sigma^1,\sigma^3)=0$ unless both $\rho^1$ and $\sigma^1$ are in $\sigma^3$ and 
then
\begin{equation}\label{M0}
\partial^{\rho^1}(\sigma^1,\sigma^3)(\underline{a})= \begin{cases}
 -\frac{1}{2\pi a 3\sqrt{2}}\quad&\mbox{if}\quad\sigma^1=\rho^1\\
 \:\:\:\frac{1}{2\pi a 3\sqrt{2}}&\mbox{if}\quad \sigma^1\neq\rho^1,\sigma^1\cap\rho^1\neq 
 \emptyset\\
  -\frac{1}{2\pi a\sqrt{2}}\quad &\mbox{if}\quad\sigma^1\cap\rho^1=\emptyset.
\end{cases}
\end{equation}
The summation over $\sigma^3$ in \eqref{var12} may be carried out using the combinatorial structure of 
$\partial\sigma^4$, see the discussion at the beginning Appendix \ref{app:3}, to give
\begin{equation}\label{M1}
M^{\rho^1,\sigma^1}=
 \begin{cases}
 -\frac{1}{2\pi a \sqrt{2}}\quad&\mbox{if}\quad\sigma^1=\rho^1\\
\:\:\:\frac{2}{2\pi a 3\sqrt{2}}&\mbox{if}\quad \sigma^1\neq\rho^1,\sigma^1\cap\rho^1\neq 
 \emptyset\\
 -\frac{1}{2\pi a\sqrt{2}}\quad &\mbox{if}\quad\sigma^1\cap\rho^1=\emptyset.
\end{cases} 
\end{equation}
In particular $M$ is a symmetric matrix.
The proof will be given in Appendix \ref{app:3}.
Thus 
\begin{equation}\label{var13}
 \dot{\delta}(\sigma^1)(\underline{a})=-\sum_{\tau^1}M^{\sigma^1\;\rho^1}(\underline{a})
 \left((\1-P(\underline{a}))u\right)_{\rho^1}.
\end{equation}
Introduce the symmetric matrices $N_1$ and $N_2$
\begin{align}\label{Ks}
N_1&=  \begin{cases}1&\mbox{if}\quad\sigma^1\cap\rho^1=\emptyset\\
                     0&\mbox{otherwise}
        \end{cases},\\
N_2&=\begin{cases} 1&\mbox{if}\quad \sigma^1\neq\rho^1,\sigma^1\cap\rho^1\neq \emptyset\\
                   0&\mbox{otherwise}
      \end{cases}.      
\end{align}
An explicit matrix representation of $N_1$ and $N_2$ will be given Appendix \ref{app:3}.

\begin{lemma}\label{lem:knauf} \cite{Knauf}
$N_1N_2=N_2N_1= 2(N_1+N_2)$ holds, so that these matrices commute.
They have the spectral decompositions
\begin{align}\label{n121}
\1 &= H_1+H_4+H_5\ \mbox{ , } \ N_1= 3H_1-2H_4+H_5\ \mbox{ , } \ 
  N_2 = 6H_1 + H_4 -2H_5\\\nonumber 
\1-P(\underline{a})&= (9\,\1-N_1-N_2)/10 = H_4 + H_5    
\end{align}
with the orthogonal projections $H_i$ to $i$--dimensional eigenspaces:
\begin{equation}\label{n122}
H_1 := (\1+N_1+N_2)/10\ \mbox{ , } \ H_4 := (6\, \1 - 4 N_1 + N_2)/15
\ \mbox{ , } \ H_5 := (3\,\1 + N_1-N_2)/6.
\end{equation}
\end{lemma}
The proof will be given in Appendix \ref{app:3} by providing an explicit matrix representation for 
$N_1,N_2$ and $\1-P(\underline{a})$.

Set 
\begin{equation}\label{var14}
 M=\frac{1}{2\pi a^{3/2}6\sqrt{2}}\widehat{M}. 
\end{equation}
With respect to a specific ordering of the $1-$simplexes and hence of the matrix indices for $M$, 
$\widehat{M}$ is given by \eqref{simplord1} in Appendix \ref{app:3}. 
Therefore with
\begin{equation}\label{var16}
\widehat{Q}_{1}=-(1-P(\underline{a}))\widehat{M}(1-P(\underline{a}))
\end{equation}
we have
\begin{equation}\label{var15}
 Q_1=\frac{1}{2\pi a^{3/2}6\sqrt{2}}\widehat{Q}_1
 \ \mbox{ with }\ \widehat{Q}_1 = 5H_4 - 10H_5.
\end{equation}
To prove the theorem, it suffices to analyze the spectrum of $Q=Q_1+Q_2$. Indeed, 
observe that $\underline{a}\in\ker Q$ and hence also
$\langle\underline{a},Q\underline{a}\rangle=0$ as predicted by Lemma \ref{lem:F}. So $0$ is an eigenvalue
of $Q$ of multiplicity at least 1. The tangent space 
$T_{\underline{a}}\cM(\partial\sigma^4,\underline{a})$ to $\cM(\partial\sigma^4,\underline{a})$ at 
$\underline{a}$, however, is just $(\1-P(\underline{a}))$. Therefore, if we can show that $Q\le 0$ and 
that
$0$ is a simple eigenvalue, then we are done. Finally, it suffices to prove this for one value of $a$ 
and we 
choose $a$ such that $2\pi a^{3/2}6\sqrt{2}=1$. 
So for the matrix 
$\widehat{Q}_1-\kappa(1-P(\underline{a}))=(5-\kappa)H_4+(-10-\kappa)H_5$, 
with $\kappa= 9\sqrt{2}\pi\left(1-\frac{3}{2\pi}\arccos\frac{1}{3}\right)=16.4846$,
we obtain its eigenvalues and their multiplicities as \cite{KK}
\begin{center}
$$
-26.4846\quad(5-\mbox{fold}),\quad -11.4846\quad (4-\mbox{fold}),
\quad  0\quad (\mbox{simple}).
$$
\end{center}
This 
shows in particular that $0$ is a simple eigenvalue.

The degeneracies of the eigenvalues in the two second variations have a simple explanation in terms of 
 representation theory. Indeed we have the following 
 \begin{theorem}\label{theo:intertw}
 Both $N_1$ and $N_2$ are intertwiners for the representation $O(s)$ of $Aut(\partial\sigma^4)$ on 
 $\R^{10}$.
 In addition $O(s)P(\underline{a})=P(\underline{a})O(s)=P(\underline{a})$ holds. 
 \end{theorem}
  \begin{proof}
 The last part is trivial. As for the first part observe that for any pair of 1-simplexes $\sigma^1$ and $\tau^1$
 and any $s$ the following is valid
   \begin{itemize}
   \item{$\sigma^1=s\tau^1$ if and only if $s^{-1}\sigma^1=\tau$}
    \item{$\sigma^1$ and $s\tau^1$ have exactly one vertex in common if and only $s^{-1}\sigma^1$ and $\tau^1$ 
    have 
    one vertex in common}
    \item{$\sigma^1$ and $s\tau^1$ have no vertex in common if and only $s^{-1}\sigma^1$ and $\tau^1$ have 
    no vertex in common.}
   \end{itemize}
The first claim then follows directly from the definitions of $N_1$ and $N_2$.
\end{proof}
\begin{corollary}\label{cor:intertw}
The spaces $\Ran H_i\;i=1,4,5$ are invariant under the representation $O(s)$. 
\end{corollary}
\begin{proof} This follows directly from \eqref{n122}.
 \end{proof}

Now decompose the representation $O(s)$ into irreducible components. By this theorem each of the second 
variations is a 
multiple of the identity transformation on each of the irreducible components. Of course $\Ran P(\underline{a})$ 
is the (only) invariant subspace for the 
trivial representation. 
\begin{lemma}
The alternating representation $s\rightarrow \sign s$ does not appear as a sub-representation of $O(s)$.
\end{lemma}
\begin{proof}
Assume there is $\underline{x}$ such that $O(s)\;\underline{x}=\sign s\; \underline{x}$ holds for all $s$. 
We will show that 
$\underline{x}=0$. Fix any $\sigma^1$. Then $x_{s^{-1}\sigma^1}=\sign s\; x_{\sigma^1}$ by the definition of $O(s)$. 
Let $s$ be the transposition of the
two vertices contained in $\sigma^1$, such that $\sign s=-1$ and $s^{-1}\sigma^1=\sigma^1$. Therefore 
$x_{\sigma^1}=0$ holds. Since $\sigma^1$ is arbitrary, this concludes the proof.
\end{proof}
As for its irreducible representations, $\bS_5$ has two one-, four-, and five- 
dimensional representations and one 6-dimensional representation. The representation matrices can be chosen such 
that their entries are integer valued, see e.g. \cite{FH}, page 28 and 60. 
Observe that $\tr O(s)=4$ holds for any transposition $s$. By comparison,
an inspection of the characters evaluated at the transpositions shows that the four-dimensional representation arising as a sub-representation of our 
$O(s)$ is the one denoted by $V$ in \cite{FH}. Similarly the five-dimensional representation arising as a 
sub-representation of $O(s)$ is the one denoted by $W$ in \cite{FH}.
This gives all irreducible components 
of $O(s)$: The trivial one- , the four-dimensional representation $V$ and the five-dimensional representation $W$, 
all appearing once. To sum up, this discussion explains the degeneracies of the two eigenvalues of the 
second variations.

This completes the proof of Theorem \ref{theo:locmax}.

\subsection{Second variation of the total scalar curvature with fixed volume.} \label{sec:volfix}

Now we will consider the variation with
\begin{equation}\label{V1}
 \underline{z}(t)=\frac{V(\underline{a})^{2/3}}
 {V(\underline{a}+t\underline{u})^{2/3}}(\underline{a}+t\underline{u}).
\end{equation}
$\underline{z}(0)=\underline{a}$ and by 
\eqref{volumescaling} $V(\underline{z}(t))=V(\underline{a})$ for all $t$. 
Set 
$$
\cM_{V(\underline{a})}(\partial\sigma^4)=\{\underline{z}\in 
\cM(\partial\sigma^4)\,|\,V(\underline{z})=V(\underline{a})\}
$$ 
and $G_{\underline{u}}(t)=\cR(\underline{z}(t))$.
In analogy to Lemma \ref{lem:F} there is 
\begin{lemma}\label{lem:G}
For any $\lambda$ the relation 
\begin{equation}\label{G1}
G_{\underline{u}+\lambda\underline{a}}(t)
=G_{\underline{u}}(t^\prime)
\end{equation}
with $t^\prime=t/(1+\lambda t)$ is valid, such that 
\begin{equation}\label{G2}
\frac{\rd^2}{\rd t^2}G_{\underline{u}+\lambda\underline{a}}(t=0)=
\frac{\rd^2}{\rd t^{\prime^2}}G_{\underline{u}}(t^\prime=0)
\end{equation}
holds. In particular
$G_{\underline{u}}(t)$ is constant if $P(\underline{a})\underline{u}=\underline{u}$ and
\begin{equation}\label{G21}
 \ddot{G}_{\underline{u}}(t=0)=\ddot{G}_{(\1-P(\underline{a}))\underline{u}}(t=0)
\end{equation}
holds for all $\underline{u}$.
\end{lemma}
\begin{proof}
\eqref{G1} follows from the trivial scaling relation
$$
\frac{1}{V(\underline{a}+t(\underline{u}+\lambda\underline{a}))}
\left(\underline{a}+t(\underline{u}+\lambda\underline{a})\right)
=\frac{1}{(1+\lambda t)^{1/2}V(\underline{a}+t^\prime \underline{u})}(\underline{a}+t^\prime\underline{u})
$$
and the scaling behavior of the Regge curvature.
In a moment we will prove
\begin{equation}\label{V3}
\dot{\underline{z}}(t=0)=(\1-P(\underline{a}))\underline{u}.
\end{equation}
Therefore the arguments in the proof of Lemma \ref{lem:F} may be 
taken over to verify the remaining claims.
\end{proof}
\begin{theorem}\label{theo:volvar}
The second variation $\ddot{G}_{\underline{u}}(t=0)$ defines an indefinite, non-degenerate quadratic form on
the tangent space $T_{\underline{a}}\cM_{V(\underline{a})}(\partial\sigma^4)$. 
Thus $(\partial\sigma^4,\underline{a})$ is a saddle point of the 
total scalar curvature when restricted to the space $\cM_{V(\underline{a})}(\partial\sigma^4)$.
\end{theorem}
Since the gradient $\underline{v}$ of the volume at $\underline{z}=\underline{a}$ 
is parallel to $\underline{a}$, the two tangent spaces 
$T_{\underline{a}}\cM_{||\underline{a}||}(\partial\sigma^4)$ and 
$T_{\underline{a}}\cM_{V(\underline{a})}(\partial\sigma^4)$ coincide.

Again for a comparison with the smooth case, see \cite{Schoen}.
\begin{proof}
\begin{equation}\label{V2}
 \dot{\underline{z}}(t)=\frac{V(\underline{a})^{2/3}}{V(\underline{a}+t\underline{u})^{2/3}}\underline{u}
 -\frac{2}{3}\frac{V(\underline{a})^{2/3}}{V(\underline{a}+t\underline{u})^{5/3}}
 \langle \underline{v}(\underline{a}+t\underline{u}),\underline{u}\rangle (\underline{a}+t\underline{u}).
\end{equation}
To establish \eqref{V3}, observe that $\underline{v}(\underline{a})=\lambda\underline{a}$ holds with 
$\lambda=\langle\underline{a},\underline{v}(\underline{a})\rangle/||\underline{a}||^2$. 
Therefore 
\begin{equation}\label{V31}
 \langle \underline{v}(\underline{a}),\underline{u}\rangle=
 \frac{\langle\underline{a},\underline{v}(\underline{a})\rangle
 \langle\underline{a},\underline{u}\rangle}{||\underline{a}||^2}=\frac{3}{2}V(\underline{a})
 \frac{\langle\underline{a},\underline{u}\rangle}{||\underline{a}||^2}.
\end{equation}
Use has been made of the Euler relation \eqref{volgrad3}. Inserting this into \eqref{V2} (with $t=0$) 
proves \eqref{V3}.

As a consequence of \eqref{V3} the first variation $\dot{G}_{\underline{u}}(t=0)$ vanishes as it 
should. Indeed,
\begin{equation}\label{V4}
\dot{G}_{\underline{u}}(t=0)=\langle\dot{\underline{z}}(t=0),\underline{Ein}(\underline{a})\rangle=
\langle (\1-P(\underline{a}))\underline{u},k\underline{a}\rangle=0.
\end{equation}
In addition 
\begin{equation}\label{V5}
\dot{|\sigma^1|}(\underline{z}(t=0))=
\frac{1}{2\sqrt{a}}\left((\1-P(\underline{a}))\underline{u}\right)_{\sigma^1}
\end{equation}
holds due to \eqref{V3}. By \eqref{var3} and \eqref{V3} $\underline{\dot{z}}(t=0)$ agree for both variations 
\eqref{var1} and \eqref{V1}. The same holds true for $|\dot{\sigma}^1|(t=0)$ by \eqref{var6} and \eqref{V5}.
Thus the first term in \eqref{var5} leads to the same quadratic form which we now denote by $Q_{1,V}$, that is 
$Q_{1,V}=Q_1$. 

The second derivative of $\underline{z}$ is
\begin{align}\label{V6}
\ddot{\underline{z}}(t)&=-
 \frac{4}{3}\frac{V(\underline{a})^{2/3}}{V(\underline{a}+t\underline{u})^{5/3}}
 \langle \underline{v}(\underline{a}+t\underline{u}),\underline{u}\rangle\underline{u}
 +\frac{10}{9}\frac{V(\underline{a})^{2/3}}{V(\underline{a}+t\underline{u})^{8/3}}
 \langle \underline{v}(\underline{a}+t\underline{u}),\underline{u}\rangle^2(\underline{a}+t\underline{u})
 \\\nonumber
&\quad-\frac{2}{3}\frac{V(\underline{a})^{2/3}}{V(\underline{a}+t\underline{u})^{5/3}} 
\langle \underline{\nabla}\,\underline{\nabla}V(\underline{a}+t\underline{u}),\underline{u}\otimes\underline{u}\rangle
(\underline{a}+t\underline{u}),
\end{align}
such that by \eqref{V31}
\begin{equation}\label{V7}
\ddot{\underline{z}}(t=0)
=-2\frac{\langle \underline{a},\underline{u}\rangle}{||\underline{a}||^2}\underline{u}
+\frac{5}{2}\frac{\langle \underline{a},\underline{u}\rangle^2}{||\underline{a}||^4}\underline{a}
-\frac{2}{3}\frac{1}{V(\underline{a})} 
\langle \underline{\nabla}\,\underline{\nabla}V(\underline{a}),\underline{u}\otimes\underline{u}\rangle
\underline{a}.
\end{equation}
The following observation allows us to shorten the calculation. By \eqref{G21} we may make the 
substitution 
$\underline{u}\rightarrow (\1-P(\underline{a}))\underline{u}$. Thus the two first terms on the r.h.s. of
\eqref{V7} vanish. The general relation \eqref{var9} then gives under this substitution
\begin{align}\label{V8}
 \ddot{|\sigma^1|}(\underline{z}(t=0))&=\ddot{\sqrt{z_{\sigma^1}}}(t=0)\\\nonumber&=-
 \frac{1}{4a^{3/2}}\left((\1-P(\underline{a}))u\right)_{\sigma^1}^2
 -\frac{2\sqrt{2}}{5 a}
 \langle (\1-P(\underline{a}))\underline{u}, M_V(\1-P(\underline{a}))\underline{u}\rangle
\end{align}
for all $\sigma^1$ and with the symmetric $10\times 10$ matrix $M_V$ given as 
\begin{equation}\label{MV1}
 M_V^{\rho^1\,\tau^1}=\partial^{\rho^1}\partial^{\tau^1}V(\underline{a})
 =\sum_{\sigma^3\in\partial\sigma^4}\partial^{\rho^1}\partial^{\tau^1}|\sigma^3|(\underline{a}).
\end{equation}
Thus we arrive at the following quadratic forms
\begin{equation}\label{V9}
\sum_{\sigma^1}\delta(\sigma^1)\ddot{|\sigma^1|}(\underline{z}(t=0))
 =\left(1-\frac{3}{2\pi}\arccos\frac{1}{3}\right)\sum_{\sigma^{1}}\ddot{|\sigma^{1}|}(\underline{a})
 =\langle\underline{u},\left(Q_{2,V}+Q_{3,V}\right)\underline{u}\rangle
 \end{equation}
 with 
 \begin{align}\label{q23}
 Q_{2,V}
 &=-\frac{1}{4a^{3/2}}\left(1-\frac{3}{2\pi}\arccos\frac{1}{3}\right)
 (\1-P(\underline{a}))
 \\\nonumber
 Q_{3,V}&=-\frac{4\sqrt{2}}{a}\left(1-\frac{3}{2\pi}\arccos\frac{1}{3}\right)
 (\1-P(\underline{a}))M_V(\1-P(\underline{a})).
 \end{align}
 \begin{lemma}\label{lem:MV}
$M_V$ is given as 
\begin{equation}\label{MV0}
M_{V}=-\frac{\sqrt{2}}{48a^{1/2}}\widehat{M}_{V,3}-\frac{2^{3/2}}{384a^{1/2}}\widehat{M}_{V,4}
\end{equation}
with 
\begin{equation}\label{MV2}
  \widehat{M}_{V,3}^{\rho^1\,\tau^1}=\begin{cases} \:\:\; 6\quad &\mbox{if}\quad \rho^1=\tau^1\\
 -2\quad &\mbox{if}  \quad \rho^1\neq\tau^1,\rho^1\cap\tau^1\neq 
 \emptyset\\
 \:\:\;3\quad &\mbox{if}\quad \rho^1\cap\tau^1=\emptyset
                         \end{cases}
\end{equation}
and
\begin{equation}\label{MV4}
 \widehat{M}_{V,4}^{\rho^1\,\tau^1}=\begin{cases} 3\quad \mbox{if}\quad \rho^1=\tau^1\\
 2 \quad \mbox{if}  \quad \rho^1\neq\tau^1,\rho^1\cap\tau^1\neq 
 \emptyset\\
 1\quad \mbox{if}\quad \rho^1\cap\tau^1=\emptyset
                         \end{cases}.
\end{equation}

\end{lemma}
The proof of this lemma will be given in Appendix \ref{app:Dervol}. With respect to the ordering 
\eqref{simplord0} of the 
$1-$simplexes $\widehat{M}_{3,V}$ and $\widehat{M}_{4,V}$ have a matrix representation given by 
\eqref{hatMatrixV} and \eqref{hatMatrixV3}.
To sum up, we have 
\begin{equation}\label{Q-V}
 Q_V=(\1-P(\underline{a}))\widetilde{M}_V(\1-P(\underline{a}))
\end{equation}
with $\widetilde{M}_V$ given as 
\begin{equation}\label{tildeMV0}
\widetilde{M}_{V}=-\gamma_1\widehat{M}-\gamma_2\;\1+\gamma_3\widehat{M}_{V,3}+\gamma_4\widehat{M}_{V,4}
\end{equation}
and where
\begin{align}\label{gammas}
\gamma_1&=\frac{1}{2\pi a^{3/2}6\sqrt{2}},\quad \hspace{2.3cm}
\gamma_2=\frac{1}{4a^{3/2}}\left(1-\frac{3}{2\pi}\arccos\frac{1}{3}\right)\\\nonumber
\gamma_3&=\frac{1}{6a^{3/2}}\left(1-\frac{3}{2\pi}\arccos\frac{1}{3}\right),
\quad\gamma_4=\frac{1}{24 a^{3/2}}\left(1-\frac{3}{2\pi}\arccos\frac{1}{3}\right).
\end{align}
Set $Q_V=\gamma_4\widetilde{Q}_V$. 
With $c:=\gamma_1/\gamma_{4}=\sqrt{2}(\pi(1- \frac{3}{2\pi}\arccos^{-1}(1/3)))^{-1}\approx 1.09193$ 
the following spectral decomposition
\begin{equation}\label{tildeq-v}
 \widetilde{Q}_V=(\1-P(\underline{a}))
 \left(-c\;\widehat{M}-6\;\1+4\;\widehat{M}_{V,3}+\widehat{M}_{V,4}\right)
 (\1-P(\underline{a}))
= (-11+5c) H_4 + (46-10c)H_5
\end{equation}
is valid.
In order to establish that $Q_V$ is indefinite for all $a$ with 0 being a simple eigenvalue, 
it suffices to show that $\widetilde{Q}_V$ has these properties. But $\widetilde{Q}_V$ has the 
following approximate eigenvalues with multiplicities \cite{KK}:
\[-5.54\quad(4-\mbox{fold}),\quad 0\quad (\mbox{simple}),\qquad 
35.08\quad (5-\mbox{fold})
\]
In particular we see again that $0$ is a simple eigenvalue. Also the interpretation of the degeneracies is as 
above, see Theorem \ref{theo:intertw}.
This concludes the proof of Theorem \ref{theo:volvar}.
\end{proof}
\section{Open Problems.}\label{sec:op}

The material provided so far gives rise to a host of open problems, of which we list some
\begin{itemize}
 \item{ Besides the examples already given find additional p.l. Einstein spaces.}
\item{In particular find p.l. Einstein metrics, which are of type I but not of type II or vice versa.}
\item{Given a pseudomanifold $K^n$, which admits an Einstein metric, are there proper subdivisions of $K^n$, 
which also admit an Einstein metric?}
\item{Given two pseudomanifolds $K_1$ and $K_2$ admitting Einstein metrics (of the same type), 
find necessary and sufficient conditions for the simplicial product $K_1\Delta K_2$ (see \cite{Spanier} 
for the definition) admitting 
an Einstein metric of the same type.} 
\item{ Compact hyperbolic manifolds are Einstein spaces. Do they have triangulations, 
which admit an Einstein metric?}
\item{Given any smooth (compact) Einstein space $\cM$, does it admit a sequence of finer and finer 
triangulations having Einstein metrics, such that the resulting p.l. Einstein spaces
approach $\cM$, e.g. in the sense of the Gromov-Hausdorff metric? For example do their total curvatures 
approach the total curvature of $\cM$, cf. \cite{CMS2}?}
\item{Can one use the concepts introduced here for interesting numerical simulations?}
\end{itemize}
For comparison recall some well known facts in the case of manifolds.
\begin{itemize}
{\item In three dimensions, $(M,g)$ is an Einstein manifold if and only if it has constant sectional curvature, 
see e.g \cite{Besse}.}
 \item {If $(M,g)$ is a four-dimensional Einstein manifold, then $\chi(M)\ge 0$ with equality 
only if $(M,g)$ is flat \cite{Berger}.}
\item{(J. Thorpe) If $(M,g)$ is a compact oriented Einstein manifold of dimension $4$, then 
$\chi(M)\ge 3/2 |\tau(M)|$ holds, where $\tau(M)$ is the signature of $M$ \cite{Thorpe,Hitchin}.}
\end{itemize}

\begin{appendix}

\section{Smooth Einstein spaces.}\label{app:Einstein}

For the purpose of making comparisons, we recall some basic and well known facts from the 
theory of Einstein spaces in Riemannian geometry, see e.g. \cite{Besse, Schoen}. In addition and for the purpose of 
comparison we shall elaborate on relations obtained from scaling the metric.

Let $M$ be a smooth, compact and closed $n-$ dimensional manifold. For any smooth Riemannian metric $g$, 
given in local coordinates
$(x^1,x^2,\cdots,x^n)$ as 
\begin{equation*}
 g(x)=\sum_{i,j}g_{ij}(x)\rd x^i\rd x^j
\end{equation*}
the volume form is 
\begin{equation*}
 \rd vol(g)(x)=\sqrt{\det g_{ij}(x)}\;\rd x^1\wedge\rd x^2\cdots\wedge\rd x^n,
\end{equation*}
the Ricci tensor is 
\begin{equation*}
 Ric(g)(x)=\sum_{i,j}Ric(g)_{ij}(x)\rd x^i\rd x^j,
\end{equation*}
and the scalar curvature is 
\begin{equation}\label{g:4}
 R(g)(x)=\sum_{i,j}g^{ij}(x)Ric(g)_{ij}(x)
\end{equation}
where $g^{ij}(x)$ is the matrix inverse to $g_{ij}(x)$. As usual, raising and lowering of indexes is achieved with 
these metric tensors. Also from now on we will use the Einstein summation convention.
The volume is
\begin{equation}
V(g)=\int_{M} \rd vol(g)(x),
\end{equation}
the {\it total scalar curvature} is
\begin{equation*}
 \cR(g)=\int_M R(g)(x)\rd vol(g)(x),
\end{equation*}
and the {\it avarage scalar curvature} is
\begin{equation*}
 \overline{R}(g)=\frac{\cR(g)}{V(g)}.
\end{equation*}
By definition $g$ is an {\it Einstein metric} and correspondingly $(M,g)$ an {\it Einstein space} 
if there exists a constant $k$ such that 
\begin{equation}\label{g:8}
 Ric(g)(x)=kg(x)
\end{equation}
holds for all $x\in M$. 
If $ g$ is an Einstein metric and if $Ric(g)(x)$ vanishes for some $x$, then trivially $k=0$ and therefore 
$Ric(g)(x)=0$ for 
all $x$, that is $(M,g)$ is Ricci-flat, compare 
Proposition \ref{prop:1} for a corresponding result in the p.l. context.

If $n\ge 3$, which we shall assume from now on, then by \eqref{g:4} necessarily $R(g)(x)$ 
is constant on $M$ - therefore equal to $\overline{R}(g)$ - 
and $k$ is given as
\begin{equation}\label{g:9}
 k=\frac{1}{n}\overline{R}(g).
\end{equation}
In general 
\begin{equation}\label{g:10}
 Ric(g)(x)-\frac{R(g)(x)}{n}g(x)
\end{equation}
is called the {\it traceless part} of $Ric(g)(x)$ and which means
\begin{equation*}
 \sum_{i,j}g^{ij}(x)\left(Ric(g)_{ij}(x)-\frac{R(g)(x)}{n}g_{ij}(x)\right)=0,
\end{equation*}
a direct consequence of \eqref{g:4}. Its integrated version 
\begin{equation}\label{g:12}
 \int_M \sum_{i,j}g^{ij}(x)\left(Ric(g)_{ij}(x)-\frac{R(g)(x)}{n}g_{ij}(x)\right)\rd vol(g)(x)=0,
\end{equation}
is of course a much weaker statement.

Given a metric $g$, the scaled metric $\lambda g$ with $\lambda>0$ is given in local coordinates by
$(\lambda g)_{ij}(x)=\lambda g_{ij}(x)$. Then trivially $(\lambda g)^{ij}(x)=\lambda^{-1}g^{ij}(x)$ holds.
If $F(g)$ is any functional of $g$, like $V(g)$ or $\cR(g)$, then $F(g)$ is said to be {\it homogeneous of degree 
$m$} if 
$F(\lambda g)=\lambda^m F(g)$ holds for all $g$. Similarly a functional $F(g)$ of $g$, which is a function on $M$, 
is homogeneous of degree $m$ if $F(\lambda g)(x)=\lambda^m F(g)(x)$ holds for all $x\in M$.
Examples are
\begin{equation}\label{g:13}
 V(\lambda g)=\lambda^{n/2} V(g),\qquad R(\lambda g)(x)
 =\lambda^{-1}R(g)(x),\quad \cR(\lambda g)=\lambda^{(n-2)/2}\cR(g).
\end{equation}

For any functional $F(g)$ its variational derivative (intuitively an infinite dimensional gradient) is written as 
\begin{equation*}
 \frac{\delta}{\delta g_{ij}(x)}F(g).
\end{equation*}
More precisely, let $h(x)=\sum_{i,j}h_{ij}(x)\rd x^i\rd x^j$ be any symmetric tensor field. Then the variational 
derivative is 
uniquely defined as a linear functional on the space of all smooth symmetric tensor fields $h$ by 
\begin{equation}\label{g:15}
 \underline{\nabla} F(g)(h)\doteq \frac{\rd}{\rd t}F(g+th)\Big|_{t=0}
 =\int_M \sum_{i,j} h_{ij}(x)\frac{\delta}{\delta g_{ij}(x)}F(g)\rd vol(g)(x).
\end{equation}
Standard examples are 
\begin{equation}\label{g:16}
 \frac{\delta}{\delta g_{ij}(x)}V(g)=\frac{1}{2}g^{ij}(x),\quad \frac{\delta}{\delta g_{ij}(x)}\cR(g)=-\left(
 Ric(g)^{ij}(x)-\frac{R(g)(x)}{2}g^{ij}(x)\right).
\end{equation}
\begin{lemma}\label{app:lem:1}
 If $F(g)$ is homogeneous of degree $m$, then 
$
\underline{\nabla}F(g)
$
 is homogeneous of degree $m-1$ and 
$$
 \frac{\delta}{\delta g_{ij}(x)}F(g)
 $$ 
is homogeneous of degree $m-n/2-1$.
\end{lemma}
\begin{proof}
We differentiate 
$$
F(\lambda g +th)=\lambda^m F(g+\frac{t}{\lambda}h)
$$
w.r.t. $t$ at $t=0$ and obtain
$$
\underline{\nabla} F(\lambda g)(h)=\lambda^{m-1}\underline{\nabla} F(g)(h),
$$
which is the first claim. As for the second, we observe that $\rd vol(g)(x)$ is homogeneous of order $n/2$, 
from which the second claim follows.
\end{proof}

$V(g)$ serves as an example. Also by\eqref{g:13} $ Ric(g)^{ij}(x)$ is homogeneous of degree $-2$ as is 
$R(g)(x)g^{ij}(x)$, see \eqref{g:10}. Therefore 
$$
Ric(g)_{ij}(x)=R(g)^{kl}(x)g_{ik}(x)g_{lj}(x)
$$ 
is homogeneous of degree $0$. 
This is compatible with \eqref{g:4}. 
\begin{corollary}\label{g:cor:1}
 If $g$ is an Einstein metric, so is $\lambda g$ for all $\lambda>0$.
\end{corollary}
The next lemma is an infinite dimensional version of Euler's relation. 
\begin{lemma}\label{app:lem:2}
 If $F(g)$ is homogeneous of degree $m$ then
 \begin{equation}\label{g:17}
 \underline{\nabla} F(g)(g)=\int_M g_{ij}(x) \frac{\delta}{\delta g_{ij}(x)}F(g)\rd vol(g)(x)=mF(g)
 \end{equation}
holds.
\end{lemma}
\begin{proof}
 Although we expect this to be well known, here is the short proof. 
For $t$ small consider $F(g+tg)=(1+t)^mF(g)$. Taking the derivative at $t=0$ gives \eqref{g:17} in view of 
\eqref{g:15}.
\end{proof}
Again $V(g)$ and $\cR(g)$ serve as examples.
Consider the functional 
\begin{equation}\label{g:18}
\widehat{\cR}(g)=\frac{1}{V(g)^{(n-2)/2}}\cR(g)=\cR\left(\frac{1}{V(g)^{2/n}}g\right),
\end{equation}
a scale invariant quantity, and observe that 
\begin{equation*}
 V\left(\frac{1}{V(g)^{2/n}}g\right)=1.
\end{equation*}
Since the Leibniz rule holds for the variational derivative, \eqref{g:15} gives 
\begin{equation}\label{g:20}
 \frac{\delta}{\delta g_{ij}(x)}\widehat{\cR}(g)
 =-\frac{1}{V(g)^{(n-2)/2}}\left(Ric(g)^{ij}(x)-\frac{R(g)(x)}{2}g^{ij}(x)
+\frac{n-2}{2n}\overline{R}(g)\,g^{ij}(x)\right).
\end{equation}
Assume now that $g$ is a critical point of $\widehat{\cR}(\cdot)$. This implies 
\begin{equation*}
 Ric(g)^{ij}(x)-\frac{R(g)(x)}{2}g^{ij}(x)+\frac{n-2}{2n}\overline{R}(g)\,g^{ij}(x)=0.
\end{equation*}
Taking the trace, see \eqref{g:4}, gives 
\begin{equation*}
 R(g)(x)-\frac{n}{2}R(g)(x)+\frac{n-2}{2}\overline{R}(g)=0,
\end{equation*}
that is the scalar curvature equals the average scalar curvature,
\begin{equation*}
 R(g)(x)=\overline{R}(g),
\end{equation*}
which when reinserted into \eqref{g:20} shows that $g$ is an Einstein metric. The converse is also true, that 
is an 
Einstein metric is a critical point of $\widehat{\cR}(g)$. There is an alternative way of defining Einstein metrics.
Consider 
\begin{equation}\label{g:24}
 A(g)=\cR(g)+\kappa V(g).
\end{equation}
In physics $\kappa$ has the interpretation of a cosmological constant.
At a critical point $g$ of $A(\cdot)$ the relation
\begin{equation}\label{g:25}
 -\left(Ric(g)^{ij}(x)-\frac{R(g)(x)}{2}g^{ij}(x)\right)+\frac{\kappa}{2} g^{ij}(x) =0
\end{equation}
holds. Taking traces gives
\begin{equation*}
 -R(g)(x)+\frac{n}{2}R(g)(x)+\frac{n}{2}\kappa=0
\end{equation*}
such that $R(g)(x)$ is constant and
\begin{equation}\label{g:27}
 \kappa=\frac{2-n}{n}R(g)(x)=\frac{2-n}{n}\overline{R}(g),
\end{equation}
which when reinserted into \eqref{g:25} shows that $g$ is an Einstein metric.

There is another way of obtaining $\kappa$ and moreover of defining an Einstein space.
Given $g$, let $L^2(M,\rd vol(g))$ be the Hilbert space of all square integrable functions on $M$ w.r.t. the 
measure
$\rd vol(g))$. The scalar product is written as $\langle\cdot, \cdot\rangle_g$.
Similarly let $\cL^2(M,\rd vol(g))$ denote the real Hilbert space of all square integrable 
symmetric tensor fields. 
That is for two such tensor fields 
$$
H=H_{ij}(x)\rd x^i\rd x^j,\qquad K=K_{ij}(x)\rd x^i\rd x^j
$$ 
the scalar product, which without risk of confusion  will also be denoted by $\langle \cdot,\cdot\rangle_g$, is 
given as 
\begin{equation*}
 \langle H, K\rangle_g=\int _M H_{ij}(x)K^{ij}(x)\rd vol(g)(x),
\end{equation*}
which indeed is positive definite, an easy consequence of the well known 
\begin{lemma}\label{app:lem2}
 Let $\M_{sym}(\R,n\times n)$ be the linear space of all real and symmetric $n\times n$ matrices and let 
 $G\in \M_{sym}(\R,n\times n)$ be positive definite. Then the real and symmetric bilinear form 
 \begin{equation*}
  \langle A,B\rangle_G= \tr (AGBG)
 \end{equation*}
on $\M_{sym}(\R,n\times n)$ is positive definite. In particular the Schwarz inequality is holds.
\end{lemma}
Thus for example 
\begin{equation*}
 \langle \1,\1\rangle_g=V(g),\quad\langle g,g\rangle_g=n V(g), \quad 
 \langle g,Ric(g)\rangle_g=\langle R(g)g,g\rangle=\cR(g),
 \end{equation*}
where $\1$ is the function on $M$ equal to 1.
We will denote by $||\;||_g$ the norms in both spaces $L^2(M,\rd vol(g))$ and $\cL^2(M,\rd vol(g))$. 
Due to \eqref{g:4} the inequality
\begin{equation}\label{g:291}
 ||R(g)||_g^2\le n||Ric(g)||_g^2
\end{equation}
is another consequence of the lemma. Since $\cR(g)=\langle R(g),\1\rangle_g$ we also have the inequality
\begin{equation}\label{g:292}
\cR(g)^2\le ||R(g)||_g^2V(g). 
\end{equation}
\begin{theorem}\label{app:theo}
 The following inequality is valid 
 \begin{equation}\label{g:30}
 \cR(g)^2\le nV(g) ||Ric(g)||_g^2
 \end{equation}
with equality if and only if $g$ is an Einstein metric and then equality in \eqref{g:291} holds.
If $(M,g)$ is an Einstein space which is not Ricci-flat, then $\kappa$ is also given as 
\begin{equation}\label{g:31} 
 \kappa=\frac{||Ric(g)||_g^2}{\cR(g)}.
\end{equation}
\end{theorem}
Observe that for an Einstein metric equality in \eqref{g:30} also follows from \eqref{g:9} and \eqref{g:31}.
\begin{proof}The first part is a direct consequence of Schwarz inequality, by which \eqref{g:30} 
is an equality if an only if $Ric$ and $g$ are collinear. Alternatively \eqref{g:30} follows by combining 
\eqref{g:291} with \eqref{g:292}.
The second part follows by taking the scalar product of \eqref{g:8} with $Ric(g)$ and the next lemma.
\end{proof}
\begin{lemma}\label{app:lem}
An Einstein space is Ricci-flat if an only if its total scalar curvature vanishes.
\end{lemma}
\begin{proof}
 If the Einstein space is Ricci-flat then obviously $\cR(g)=0$. As for the converse, if $\cR(g)=0$, 
 then $\kappa=0$ by \eqref{g:27} and hence the Ricci tensor vanishes.
\end{proof}
\begin{corollary}\label{app:corr}
 Equality in \eqref{g:291} holds if and only if for all $x$ there is $\kappa(x)$ such that
$Ric(g)(x)=\kappa(x)g(x)$ holds. Equality in \eqref{g:30} implies equality in \eqref{g:291}.
\end{corollary}
\begin{proof}
 If $Ric(g)(x)=\kappa(x)g(x)$ holds for all $x$ with a suitable $\kappa(x)$ then \eqref{g:291} holds. Conversely 
 assume 
 \eqref{g:291} holds. Then for almost all $x$ there is $\kappa(x)$ such that $Ric(g)(x)=\kappa(x)g(x)$ holds. 
 But then $\kappa(x)=R(g)(x)/n$ for these $x$ and by continuity we can make this relation hold for all $x$.
 The last claim is now obvious.
 \end{proof}
Again for comparison we conclude with recalling Hamilton's Ricci flow equations.
The unnormalized flow equation for the metric is defined as
\begin{equation*}
 \frac{\rd}{\rd t}g(t)_{ij}(x)=-2Ric(g(t))_{ij}(x)
\end{equation*}
while the normalized one is given as 
\begin{equation}\label{h2}
 \frac{\rd}{\rd t}g(t)_{ij}(x)=-2Ric(g(t))_{ij}(x)+\frac{2}{n}\;\overline{R}(g(t))\;g(t)_{ij}(x).
\end{equation}
Under the normalized flow the volume is conserved. This follows easily by taking the derivative 
of $V(g(t))$ with help the first relation in \eqref{g:16}, the flow equation and \eqref{g:4}. 
Also observe that the r.h.s. of \eqref{h2} vanishes, if 
$g(t)$ is an Einstein metric. In other words, any Einstein metric is a fixed point of the 
normalized flow equation.
Theorem \ref{app:theo} and in particular relation \eqref{g:31} suggests another normalized Ricci flow.
\begin{equation}\label{h3}
 \frac{\rd}{\rd t}g(t)_{ij}(x)=-2Ric(g(t))_{ij}(x)+2\frac{||Ric(g(t))||_{g(t)}^2}{\cR(g(t))}\;\;g(t)_{ij}(x).
\end{equation}
which is well defined as long as $\cR(g(t))\neq 0$. 
By the previous theorem, any Einstein metric is a fixed point. 
Although believed to be known, the author has not been able to locate a reference for the next result.
\begin{theorem}\label{app:prop}
Under the flow \eqref{h3} the volume $V(g(t))$ increases if $\,\cR(g(t))>0$ and decreases if 
$\,\cR(g(t))<0$, while the total scalar curvature itself increases in both cases as long as $\cR(g(t))\neq 0$.
\end{theorem}
\begin{proof}
We calculate
\begin{align*}
\frac{\rd}{\rd t} V(g(t))&=\int g(t)^{ij}\left(-Ric(g(t))_{ij}
+\frac{||Ric(g(t))||_{g(t)}^2}{\cR(g(t))}g(t)_{ij}\right)\rd vol(g(t))\\ 
&=-\frac{1}{\cR(g(t))}\left( \cR(g(t))^2-nV(g(t)) ||Ric(g(t))||_{g(t)}^2\right)
\end{align*}
and the first claim follows from \eqref{g:30}.
The second claim follows from
\begin{align*}
\frac{\rd}{\rd t} \cR(g(t))&=\langle \dot{g}(t),
-\left(Ric(g(t))-\frac{1}{2}R(g(t))g(t)\right)\rangle_{g(t)}\\\nonumber
&=\langle -2 Ric(g(t))+2\frac{||Ric(g(t))||_{g(t)}^2}{\cR(g(t))}g(t),-\left(Ric(g(t))
-\frac{1}{2}R(g(t))g(t)\right)\rangle_{g(t)}\\\nonumber
&=-||R(g(t))||^2_{g(t)}+n||Ric(g(t))||_{g(t)}^2
\end{align*}
and \eqref{g:291}.
\end{proof}

 \section{Proofs of the relations \eqref{einsteincurv} -- \eqref{einsteinrvol}.}\label{app:1}

In $\partial \sigma^{n+1}$ any $(n-2)$-simplex is the face of 
$3\:\:n$-simplexes. 
So in units of $2\pi$, by \eqref{dihedralangle} the deficit angle at any $(n-2)-$simplex is given by 
\begin{equation*}
 \delta(n)=\left(1-\frac{3}{2\pi}\arccos\left(\frac{1}{n}\right)\right).
\end{equation*}
Now $\delta(n)$
is a monotonically decreasing function of $n$ with limiting value $1/4$
as $n\rightarrow\infty$. Its value for $n=2$ is $1/2$.
Thus $\delta(n)$ is strictly positive.
Also there are a total of 
\begin{equation*}
 \left(\begin{matrix}
 n+2\\n-1
\end{matrix}\right)
\end{equation*}
$(n-2)-$simplexes in $\partial \sigma^{n+1}$. Collecting terms and using \eqref{Voleinstein} for the 
volume of an equilateral $(n-2)-$simplex gives the total scalar curvature \eqref{einsteincurv}. 
Because $\delta(n)$ is strictly positive so is the total scalar curvature.
\eqref{einsteinric} then follows by using Euler's relation and the fact that $Ein_{\sigma^1}$ is independent of 
$\sigma^1$. Since there are 
$$
\left(\begin{matrix}
 n-1\\2
\end{matrix}\right)
$$
$1-$simplexes in an $(n-2)-$simplex, again by Euler's relation 
$$
\partial^{\sigma^1}|\sigma^{n-2}|(\underline{a})=\frac{1}{a \left(\begin{matrix}
 n-1\\2
\end{matrix}\right)}\frac{n-2}{2}|\sigma^{n-2}|(\underline{a})=
\frac{1}{a(n-1)}|\sigma^{n-2}|(\underline{a})
$$
when $\sigma^1\in\sigma^{n-2}$ and zero otherwise. In particular 
$
\partial^{\sigma^1}|\rho^1|(\underline{a})=
\delta^{\sigma^1\,\rho^1}\cdot 1/2\sqrt{a}
$
as it should.
Using \eqref{Voleinstein} gives \eqref{partialvol}.



\section{Proof of Lemma \ref{lem:3}.}\label{app:2}

By iteration the relation \eqref{dervol} implies that each volume $|\sigma^k(\underline{z})|$ is a
smooth function in $\underline{z}$.  Thus it suffices to show that each dihedral angle $(\sigma^{n-2},\sigma^n)$
is also smooth in $\underline{z}$.  
As in the proof of Theorem \eqref{theo:5} $v_1,\cdots,v_n$ denotes an ordered basis in $E^n$. It defines a 
euclidean $n-$simplex $\sigma^n$, the convex hull the origin and the $v_1,\cdots,v_n$, which thus are the 
vertices. 
The edge lengths squared are 
the $||v_i||^2$ and the $||v_i-v_k||^2,\;k<i$. By the simple polarization formula
\begin{equation}\label{a2:1}
 \langle v_i,v_k\rangle=\frac{1}{2}\left( ||v_i||^2+||v_k||^2 -\langle v_i-v_k,v_i-v_k\rangle\right)
\end{equation}
all these scalar products are expressible in terms of the edge lengths squared.
Let $\Lambda^l(E^n)$ denote the $l$-th exterior power of $E^n$. The inner product on this space is given by
\begin{equation}\label{a2:2}
 \langle x_1\wedge \cdots \wedge x_l,y_1\wedge\cdots \wedge y_l\rangle =\det \langle x_i,y_k\rangle.
\end{equation}
In particular the volume of the euclidean simplex $\sigma^n$ equals
$$
|\sigma^n|=\frac{1}{n!}||v_1\wedge\cdots\wedge v_n||.
$$
Set 
\begin{equation*}
 0\neq w_i=(-1)^i v_1\wedge\cdots\wedge \widehat{v}_i\wedge v_l\quad \in\quad \Lambda^{n-1}(E^n).
\end{equation*}
By \eqref{a2:1} and \eqref{a2:2} the $\langle w_i,w_k\rangle$ are polynomials in the edge lengths squared.
This has the following consequence. Let $\Theta_{ij}$ be the angle (normalized to $2\pi$) of the two 
hyperplanes determined by $w_i$ and $w_k$. Then 
\begin{equation*}
 \Theta_{ik}=1-\frac{1}{\pi}\arccos\frac{\langle w_i,w_k\rangle}{||w_i||\,||w_k||}.
\end{equation*}
In fact $\Theta_{ij}$ is the dihedral angle at the $(n-2)$-simplex, which is the convex hull of the origin and the 
$v_1,\cdots,\cdots, \widehat{v}_k,\cdots,\widehat{v}_i,\cdots, v_n$. In particular we conclude that 
$\Theta_{ij}$ is smooth in the edge lengths squared. The smoothness in $\underline{z}$ of the 
dihedral angles at the remaining $(n-2)$-simplexes -- each of them is the convex hull of the 
$v_1,\cdots,\widehat{v}_i,\cdots, v_n$ for a suitable $i$ -- may be established similarly. 
This concludes the proof of Lemma \ref{lem:3}.


\section{Proof of Relation \eqref{M1} and of Lemma \ref{lem:knauf} }\label{app:3}
We start with the proof of the lemma.

As for the proof of \eqref{M1} we start with some observations on the combinatorial structure of 
$\partial\sigma^4$, which has
five vertices, ten $1$-simplexes, ten $2$-simplexes and five $3$-simplexes.

Given two $1$-simplexes $\sigma^1$ and $\tau^1$ in $\partial \sigma^4$, 
we will distinguish three cases concerning the $3$-simplexes they are contained in.
\label{comb}
\begin{enumerate}
 \item{If $\sigma^1=\tau^1$, then both are contained in exactly three $3$-simplexes}
 \item{If $\sigma^1$ and $\tau^1$ have exactly one vertex in common, then both are contained in 
 exactly two $3$-simplexes.}
 \item{If $\sigma^1$ and $\tau^1$ have no vertex in common, then both are contained in exactly one $3$-simplex.}
\end{enumerate}

Also, if $\sigma^1\in\sigma^3$, then there is exactly one $1$-simplex, denoted by 
$\bar{\sigma}^1\in \sigma^3$, such that $\sigma^1\in\sigma^3$ and $\overline{\sigma}^1\in \sigma^3$ 
have no vertex in common. Finally any $1$-simplex is contained in exactly three $3$-simplexes.
Also for given $1$-simplex $\sigma^1$ there are six different $1$-simplexes, which have exactly one vertex
in common with $\sigma^1$ and three $1$-simplexes, which have no vertex in common with $\sigma^1$.
This agrees of course with the fact, that altogether there are ten $1$-simplexes in $\partial\sigma^4$.

With these preparations and taking the symmetry of $\partial \sigma^4$ into account 
it suffices to calculate
\begin{equation*}
\partial^{\rho^1}(\sigma^1,\sigma^3)(\underline{a}).
\end{equation*}
We remark, that there is formula, which expresses the dihedral angles at any euclidean tetrahedron in terms of 
its edge length, see \cite{Lee}, Proposition 3.1. However, we will follow a different approach.
Of course if $\sigma^1\nsubseteq \tau^3$ then this expression vanishes. So it suffices to consider 
a single $3-$simplex.
Set 
\begin{align*}
\partial^{\rho^1}(\sigma^1,\sigma^3)(\underline{a})&=\alpha,\\\nonumber
\partial^{\rho^1}(\sigma^1,\sigma^3)(\underline{a})&=\beta,\quad \sigma^1\neq \rho^1,
\sigma^1\cap \rho^1 \neq \emptyset,\\\nonumber
\partial^{\rho^1}(\sigma^1,\sigma^3)(\underline{a})&=\gamma,\quad \sigma^1\cap\rho^1=\emptyset,
\end{align*}
for $\sigma^1,\rho^1\in\sigma^3$.

In order to calculate $\alpha,\beta$ and $\gamma$, 
consider the euclidean 3-simplex of which five edges have length $\sqrt{a}$, while the 
remaining one has edge length $\sqrt{x}$ with $0\le x<3a$. 
Denote the vertices by $v_0,v_1,v_2, v_3$. The vertex
$v_0$ is located at the origin. The three other ones have the coordinates
\begin{align*}
 v_1&=\left(\sqrt{a},0,0\right)\\\nonumber
 v_2&=\left(\frac{1}{2}\sqrt{a}, \frac{1}{2}\sqrt{3a},0\right)\\\nonumber
 v_3&=\left(\frac{x}{2\sqrt{a}},\frac{x}{2\sqrt{3a}},\sqrt{x\left(1-\frac{x}{3a}\right)}\;\right).
\end{align*}
We calculate the outward unit normal vectors to the four faces.
They are 
\begin{align*}
 n_1&=-\frac{v_1\times v_2}{||v_1\times v_2||},\quad
 n_2=\quad\frac{v_1\times v_3}{||v_1\times v_3||},\\\nonumber
 n_3&=-\frac{v_2\times v_3}{||v_1\times v_3||},\quad
 n_4=-\frac{(v_3-v_1)\times (v_2-v_1)}{||(v_3-v_1)\times(v_2-v_1)||},
\end{align*}
with $\times$ denoting the vector product.
A straight forward calculation gives 
 \begin{align*}
\langle n_1,n_2\rangle(x)&=
\langle n_1,n_3\rangle(x)=\langle n_2,n_4\rangle=\langle n_3,n_4\rangle(x)\\\nonumber
&=-\frac{x}{2\sqrt{3}}\frac{1}{\sqrt{ax-\frac{x^2}{4}}}, \\\nonumber
\langle n_1,n_4\rangle(x)&=\frac{2}{a\sqrt{3}}
\left(\frac{x}{\sqrt{3}}-\frac{a}{2}\sqrt{3}\right)\\\nonumber
\langle n_2,n_3\rangle(x)&= -\frac{1}{2\left(ax-\frac{x^2}{4}\right)}  
 \left(ax-\frac{x^2}{2}\right).
\end{align*}
The equality of $\langle n_1,n_2\rangle,\langle n_1,n_3\rangle,\langle n_2,n_4\rangle$ and 
$\langle n_3,n_4\rangle$ follows also from symmetry considerations. 
In agreement with \eqref{dihedralangle} the relations 
$\langle n_i,n_j\rangle(x=a)=-1/3, i\neq j$ hold.
Consider the function $y=\pi-\arccos f(x)=\arccos(-f(x))$ whose derivative is given as
\begin{equation}\label{app62}
\frac{\rd y}{\rd x}=\frac{1}{\sqrt{1-f(x)^2}}\frac{\rd f(x)}{\rd x},
\end{equation}
as long as $-1\le f(x)\le 0$ and correspondingly $0\le y\le \pi$. 
In what follows $f$ will be one of the three quantities
$\langle n_2,n_4\rangle,\langle n_2,n_3\rangle$ and $\langle n_1,n_3\rangle$. In particular 
$f(a)=-1/3$ such that $\sqrt{1-f(a)^2}=2^{3/2}/3$. Also $y$ will be one of the six 
dihedral angles. Indeed, an easy argument shows that the angle between two normals and the 
corresponding dihedral angle sum up to $\pi$. Therefore we obtain for the derivatives of 
the scalar products of the normals and thus for the derivatives of the dihedral angles the expressions
\eqref{app62}
\begin{align*}
 \alpha&=\frac{1}{2\pi}\left(\frac{3}{2\sqrt{2}}\right)
 \frac{\rd}{\rd x}\langle n_2,n_3 \rangle(x=a)\\\nonumber
 \beta&=\frac{1}{2\pi}\left(\frac{3}{2\sqrt{2}}\right)
 \frac{\rd}{\rd x}\langle n_1,n_2 \rangle(x=a)\\\nonumber
 \gamma&=\frac{1}{2\pi}\left(\frac{3}{2\sqrt{2}}\right)
 \frac{\rd}{\rd x}\langle n_1,n_4\rangle(x=a)
\end{align*}
in units of $2\pi$. A short calculation gives
\begin{equation*}
\alpha=\frac{1}{2\pi a 3\sqrt{2}},\quad \beta=-\frac{1}{2\pi a 3\sqrt{2}}=-\alpha,
\quad\gamma=\frac{1}{2\pi a\sqrt{2}}=3\alpha
\end{equation*}
and the claim \eqref{M0} follows.
We turn to a proof of Lemma \ref{lem:knauf}.
Give the five vertices of $\partial \sigma^4$ the labels $0,\cdots, 4$ and accordingly write 
the ten $1-$ simplexes ordered in terms of the two vertices in their boundary as 
\begin{align}\label{simplord0}
\sigma_1^1&=\sigma^1_{01},\sigma^1_2=\sigma^1_{02},\sigma^1_3=\sigma^1_{03}, 
\sigma^1_4=\sigma^1_{04},\sigma^1_5=\sigma^1_{12},\\\nonumber
\sigma^1_6&=\sigma^1_{13},
\sigma^1_7=\sigma^1_{14},\sigma^1_8=\sigma^1_{23},\sigma^1_9=\sigma^1_{24},\sigma^1_{10}=\sigma^1_{34}.
\end{align}
With this ordering of the $1$-simplexes the matrices $N_1$ and $N_2$ take the form
\begin{equation*}
N_1=\left(
\begin{smallmatrix}
 0 & 0 & 0 & 0 & 0 & 0 & 0 & 1 & 1 & 1 \\
 0 & 0 & 0 & 0 & 0 & 1 & 1 & 0 & 0 & 1 \\
 0 & 0 & 0 & 0 & 1 & 0 & 1 & 0 & 1 & 0 \\
 0 & 0 & 0 & 0 & 1 & 1 & 0 & 1 & 0 & 0 \\
 0 & 0 & 1 & 1 & 0 & 0 & 0 & 0 & 0 & 1 \\
 0 & 1 & 0 & 1 & 0 & 0 & 0 & 0 & 1 & 0 \\
 0 & 1 & 1 & 0 & 0 & 0 & 0 & 1 & 0 & 0 \\
 1 & 0 & 0 & 1 & 0 & 0 & 1 & 0 & 0 & 0 \\
 1 & 0 & 1 & 0 & 0 & 1 & 0 & 0 & 0 & 0 \\
 1 & 1 & 0 & 0 & 1 & 0 & 0 & 0 & 0 & 0 \\
\end{smallmatrix}\right)
\ \mbox{ and } \ 
N_2 =\left(
\begin{smallmatrix}
 0 & 1 & 1 & 1 & 1 & 1 & 1 & 0 & 0 & 0 \\
 1 & 0 & 1 & 1 & 1 & 0 & 0 & 1 & 1 & 0 \\
 1 & 1 & 0 & 1 & 0 & 1 & 0 & 1 & 0 & 1 \\
 1 & 1 & 1 & 0 & 0 & 0 & 1 & 0 & 1 & 1 \\
 1 & 1 & 0 & 0 & 0 & 1 & 1 & 1 & 1 & 0 \\
 1 & 0 & 1 & 0 & 1 & 0 & 1 & 1 & 0 & 1 \\
 1 & 0 & 0 & 1 & 1 & 1 & 0 & 0 & 1 & 1 \\
 0 & 1 & 1 & 0 & 1 & 1 & 0 & 0 & 1 & 1 \\
 0 & 1 & 0 & 1 & 1 & 0 & 1 & 1 & 0 & 1 \\
 0 & 0 & 1 & 1 & 0 & 1 & 1 & 1 & 1 & 0 \\
\end{smallmatrix}
\right).
\end{equation*}
Also 
\begin{equation*}
 P(\underline{a})=\frac{1}{10}\left(
\begin{smallmatrix}
 1 & 1 & 1 & 1 & 1 & 1 & 1 & 1 & 1 & 1 \\
 1 & 1 & 1 & 1 & 1 & 1 & 1 & 1 & 1 & 1 \\
 1 & 1 & 1 & 1 & 1 & 1 & 1 & 1 & 1 & 1 \\
 1 & 1 & 1 & 1 & 1 & 1 & 1 & 1 & 1 & 1 \\
 1 & 1 & 1 & 1 & 1 & 1 & 1 & 1 & 1 & 1 \\
 1 & 1 & 1 & 1 & 1 & 1 & 1 & 1 & 1 & 1 \\
 1 & 1 & 1 & 1 & 1 & 1 & 1 & 1 & 1 & 1 \\
 1 & 1 & 1 & 1 & 1 & 1 & 1 & 1 & 1 & 1 \\
 1 & 1 & 1 & 1 & 1 & 1 & 1 & 1 & 1 & 1 \\
 1 & 1 & 1 & 1 & 1 & 1 & 1 & 1 & 1 & 1\\
\end{smallmatrix}
\right).                        
\end{equation*}
Lemma \ref{lem:knauf} follows by an easy calculation.

By \eqref{M1} and the definition \eqref{var14} of $\widehat{M}$ 
\begin{equation}\label{simplord1}
\widehat{M} = 3\,\1 + 3N_1 - 2N_2 = -5H_4 + 10H_5
\end{equation}
follows. Similarly 
\begin{equation}\label{hatMatrixV}
\widehat{M}_{V,3} = 6\,\1 + 3N_1 - 2N_2 = 3H_1 - 2H_4 + 13H_5
\end{equation}
and
\begin{equation}\label{hatMatrixV3}
\widehat{M}_{V,4} =  3\,\1 + N_1 + 2N_2 = 18H_1 + 3H_4.
\end{equation}

\section{Proof of relations \ref{MV2} and \ref{MV4}. }\label{app:Dervol}

Recall that in the proof of \eqref{M1} use was made of the symmetry of the boundary 
$\partial\sigma^{3+1}$ of the simplex $\sigma^{3+1}$. This applies 
here too for $M_V$, so apart from combinatorial counting the main 
calculation to be done is to determine the partial derivatives up to order two of the volume of a single 
$3$-simplex at its equilateral value, that 
is $\partial^{\sigma^1}\partial^{\tau^1}|\sigma^3|(\underline{a})$. 
Label the four vertices of 
$\sigma^3$ as 
$0,1,2,3$. Correspondingly write the six $1$-simplexes as $\{01\},\{02\},\{0,3\},\{12\},\{13\},\{23\}$ and 
the six lengths squares as $z_{01},z_{02},z_{03},z_{12},z_{13},z_{23}$.
Consider the symmetric $3\times 3$ matrix $A(\underline{z})$, see \eqref{amatrix},
\begin{equation*}
 A(\underline{z})=\left(\begin{array}{ccc}
 z_{01}&\frac{1}{2}(z_{01}+z_{02}-z_{12})&\frac{1}{2}(z_{13}-z_{01}-z_{03})\\
 &&\\
 \frac{1}{2}(z_{12}-z_{01}-z_{02})&z_{02}&\frac{1}{2}(z_{23}-z_{02}-z_{03})\\
 &&\\
 \frac{1}{2}(z_{13}-z_{01}-z_{03})&\frac{1}{2}(z_{23}-z_{02}-z_{03})&z_{03}
 \end{array}
 \right)
\end{equation*}
The volume $|\sigma^3|(\underline{z})$ can be obtained from $A(\underline{z})$, see \eqref{amatrixvol}, in this case
\begin{equation*}
 |\sigma^3|(\underline{z})=\frac{1}{6}\sqrt{\det A(\underline{z})}
\end{equation*}
giving 
\begin{equation*}
 \partial^{\rho^1}\partial^{\tau^1}|\sigma^3|(\underline{z})=
 \frac{1}{12}\frac{\partial^{\rho^1}\partial^{\tau^1}\det A(\underline{z})}{\det A(\underline{z})^{1/2}}
 -\frac{1}{24}\frac{\partial^{\rho^1}\det A(\underline{z})\partial^{\tau^1}\det A(\underline{z})}
 {\det A(\underline{z})^{3/2}}.
\end{equation*}
So it suffices to calculate the 
partial derivatives of $\det A(\underline{z})$ up to order two. 
An easy computer supported calculation gives 
\begin{align*}
\det A(\underline{z})&=-\frac{1}{4}z_{01}^2 z_{23}
-\frac{1}{4}z_{02}^2 z_{13}-\frac{1}{4}z_{03}^2 z_{12}
-\frac{1}{4}z_{12}^2 z_{03}-\frac{1}{4}z_{13}^2 z_{02}-\frac{1}{4}z_{23}^2 z_{01}\\\nonumber
&\quad -\frac{1}{4}z_{01}z_{02}z_{12}-\frac{1}{4}z_{01}z_{03}z_{13}
+\frac{1}{4}z_{01}z_{02}z_{23}\\\nonumber
&\quad+\frac{1}{4}z_{01}z_{03}z_{12}+\frac{1}{4}z_{01}z_{02}z_{13}
+\frac{1}{4}z_{01}z_{13}z_{23}\\\nonumber
&\quad+\frac{1}{4}z_{01}z_{03}z_{23}+\frac{1}{4}z_{01}z_{12}z_{23}
+\frac{1}{4}z_{02}z_{03}z_{13}\\\nonumber
&\quad+\frac{1}{4}z_{02}z_{03}z_{12}+\frac{1}{4}z_{02}z_{12}z_{13}
-\frac{1}{4}z_{02}z_{03}z_{23}
+\frac{1}{4}z_{02}z_{13}z_{23}\\\nonumber
&\quad +\frac{1}{4}z_{03}z_{12}z_{13}+\frac{1}{4}z_{03}z_{12}z_{23}
-\frac{1}{4}z_{12}z_{13}z_{23}.
\end{align*}
This gives $\det A(\underline{a})=a^3/2$ and the volume as 
$|\sigma^3|(\underline{a})=a^{3/2}/6\sqrt{2}$
in agreement with the general formula \eqref{Voleinstein}. 
As another consequence 
\begin{equation*}
\partial^{\sigma^1}\det A(\underline{a})=\frac{1}{4}a^2,
\quad\mbox{for all} \quad\sigma^1.
\end{equation*}
Next come second order partial derivatives
\begin{align*}
\partial^{\sigma^1}\partial^{\sigma^1}\det A(\underline{a})&=-\frac{1}{2}a,\quad \mbox{for all}\quad \sigma^1
\\\nonumber
\partial^{\{01\}}\partial^{\{23\}}\det A(\underline{a})&=
\partial^{\{02\}}\partial^{\{13\}}\det A(\underline{a})= 
\partial^{\{03\}}\partial^{\{12\}}\det A(\underline{a})=-\frac{3}{4}a,\\\nonumber
\partial^{\sigma^1}\partial^{\tau^1}\det A(\underline{a})&=\frac{1}{4}a,\quad\mbox{for all other}\quad \sigma^1,\tau^1,
\end{align*}
again with equalities as required by symmetry.
Combining this result with \eqref{MV1} and the combinatorial structure of $\partial\sigma^4$ - as discussed 
at the beginning of Appendix \ref{app:3} - the claims \eqref{MV2} and \eqref{MV4} follow by a short calculation.
\end{appendix}


\end{document}